\UseRawInputEncoding
\documentclass[12pt,a4paper,english]{article}
\usepackage[utf8]{inputenc}
\usepackage[english]{babel} 
\usepackage[T1]{fontenc}        
\usepackage{textcomp}

\usepackage{amsmath,amsthm,amssymb,amsfonts,enumitem}
\usepackage{bm}
\usepackage{blkarray}

\usepackage{setspace}
\usepackage{hyperref}

\usepackage{caption}
\usepackage{subcaption}

\usepackage{natbib}
\usepackage{booktabs, siunitx}

\usepackage{pdfpages}
\usepackage{graphicx}

\newtheorem{definition}{Definition}
\newtheorem{theorem}{Theorem}[section]

\newtheorem{lemma}[theorem]{Lemma}
\newcounter{example}[section]

\usepackage[section]{placeins}

\newcommand{\N}{\mathcal{N}}
\newcommand{\A}{\mathcal{A}}
\newcommand{\R}{\mathbb{R}}
\newcommand{\E}{\mathbb{E}}
\newcommand{\Z}{\mathbb{Z}}
\newcommand{\Var}{\mathbb{V}}

\newcommand{\AV}{\text{AV }}

\newcommand{\Y}{\mathcal{Y}}
\newcommand{\BY}{\bm{Y}}

\newcommand{\B}{\mathcal{B}_0}
\newcommand{\U}{\mathcal{U}}
\newcommand{\V}{\mathcal{V}}
\newcommand{\PP}{\mathcal{P}}

\newcommand{\Bbeta}{\bm{\beta}}
\newcommand{\Bb}{\bm{b}}
\newcommand{\By}{\bm{y}}
\newcommand{\Btheta}{\bm{\theta}}

\newcommand{\Cov}{\mathrm{Cov}}
\newcommand{\deffeq}{\mathrel{\overset{\makebox[0pt]{\mbox{\normalfont\tiny\sffamily def}}}{=}}}
\newcommand{\tth}{\footnotesize ^{\mbox{th}}\normalsize}
\newcommand{\ie}{{\it i.e., }}
\newcommand{\eg}{{\it e.g., }}

\DeclareMathOperator*{\argmax}{arg\,max}

\usepackage{authblk}

\title{Conditional Inference for Multivariate Generalised Linear Mixed Models}
\author[1]{Jeanett S. Pelck}
\author[1*]{Rodrigo Labouriau}
\affil[1]{Department of Mathematics,  Aarhus University, Denmark}
\affil[*]{\footnotesize Corresponding author: Rodrigo Labouriau, rodrigo.labouriau@math.au.dk \normalsize}

\date{July 2021}

\begin{document}

\maketitle

\begin{abstract}

\noindent
We propose a method for inference in generalised linear mixed 
models (GLMMs) and several extensions of these models. First, we 
extend the GLMM by allowing the distribution of the random 
components to be non-Gaussian, that is, assuming an absolutely 
continuous distribution with respect to the Lebesgue measure that is 
symmetric around zero, unimodal and with finite moments up to 
fourth-order. Second, we allow the conditional distribution to follow 
a dispersion model instead of exponential dispersion models. Finally, 
we extend these models to a multivariate framework where multiple 
responses are combined by imposing a multivariate absolute 
continuous distribution on the random components representing common clusters of observations in all the marginal models.

Maximum likelihood inference in these models involves evaluating an integral that often cannot be computed in closed form. We suggest an inference method that predicts values of  random components and does not involve the integration of conditional likelihood quantities.
% while preserving desirable properties of the likelihood-based inference.
% These predictions are sufficient statistics with respect to the covariance matrices in the multivariate distributions assumed for the random components, and ancillary statistics with respect to the scale and location parameters in the dispersion models. 
% Using the theory of inference functions, we show that for small variances of the random components, the estimated fixed effects are asymptotically Gaussian distributed. 
The multivariate GLMMs that we studied can be constructed with marginal 
GLMMs of different statistical nature, and at the same time, represent complex 
dependence structure providing a rather flexible tool for applications. 
\end{abstract}

\newpage

\tableofcontents

\newpage

\section{Introduction}

Generalised linear mixed models (GLMMs) form a flexible class of statistical 
models, which combines the capability to incorporate non-Gaussian distributions and 
non-linear link functions, inherited from standard generalised linear models, with 
the power of representing complex dependence structures using random 
components in the same fashion as classic (Gaussian) mixed models. Therefore, 
GLMMs appear as a natural tool in many applications (see 
\citeauthor{Demidenko2013}, \citeyear{Demidenko2013}; \citeauthor{Mcculloch2001}, \citeyear{Mcculloch2001}; \citeauthor{Fahrmeir2001}, \citeyear{Fahrmeir2001} and 
\citeauthor{Agresti2002}, \citeyear{Agresti2002}).
%
% Demidenko, Eugene. Mixed Models: Theory and Applications with R (Wiley Series in Probability and Statistics) (p. 699). Wiley. Kindle Edition. 
%
% McCulloch, C.E. and Searle, S.R. (2001). Generalized, Linear and Mixed Models. New York: Wiley.
%
% Fahrmeir, L. and Tutz, G (2001). Multivariate Statistical Models based on Generalized Linear Mixed Models. Springer Verlag. New York.
%
% Agresti, A. (2002). Categorical Data Analysis. New York: Wiley.
% 
However, the power of GLMMs comes with a price: the required inference tools 
are more demanding than standard statistical models. For instance, the 
likelihood-based inference requires a non-trivial integration of conditional 
likelihood quantities. Moreover, some of the simplifications of the integration 
used in the classic Gaussian mixed models (\eg the result of conditioning a 
Gaussian distribution on Gaussian random components yields a Gaussian marginal 
distribution) do not apply in general for GLMMs. For this reason, several 
inferential tools are discussed in the literature; see \cite{Breslow1993}, 
\cite{Mcculloch2001}; see also \cite{McCulloch1997} for a comprehensive study comparing several methods 
ranging from simple numeric (quadrature) integration of the conditional 
likelihood to several versions of the EM algorithm. 

In this paper we present an alternative method of inference for GLMMs, 
constructed using inference functions, which avoids integrating likelihood 
quantities while preserving some of the desirable properties of classic 
likelihood-based methods. Moreover, this new method applies to GLMMs 
with minimal requirements for the distribution of the random components, which 
are not necessarily assumed to be normally distributed, as in the standard setup of 
GLMMs. For instance, we will be able to consider models with 
heavy-tailed random components as the multivariate t-distribution. 

The methods we expose allow us to construct natural extensions to multivariate 
GLMMs. The main idea is to construct one GLMM describing each response. It is 
assumed that there is a natural cluster of observations (\eg individuals or 
experimental units). Each of those GLMMs contains random components 
representing those clusters, \ie taking the same value for all the observations 
belonging to the same cluster. The multivariate GLMM is then constructed by 
assuming that the distributions of the random components representing the 
clusters are the marginal distributions of a multivariate distribution (\eg a 
multivariate normal distribution or a multivariate t-distribution).
Note that the multivariate generalised linear mixed models (MGLMMs), that we 
obtain in this way, can have marginal models of different nature which might be defined with different distributions and different link functions. 
In this way, those multivariate models can simultaneously describe responses of varying nature in a way that is not possible do with classic multivariate Gaussian models. 
Furthermore, since we defined the random components of the marginal GLMMs using minimal distributional assumptions, we will also obtain a MGLMM constructed with a flexible class of multivariate random components. For instance, the multivariate random 
components can be multivariate normally distributed or regular elliptical 
contoured distributed.

The paper is structured as follows. In Section \ref{Sec:2}, we introduce an 
extension of GLMMs constructed using random components that are not normally 
distributed, and by extending the family of conditional distributions.  We use a 
simple case, containing random components representing a grouping of the 
observations (denoted clusters) due to the observational scheme used in the 
experiment,
to present the ideas behind the inference techniques we propose in Section 
\ref{SubSec:2.1}, and expose the basic asymptotic properties of those techniques 
in Section \ref{SubSec:2.2}. Section \ref{SubSec:2.3} extends the inference 
techniques to the case of models with complex clustering structures. In Section 
\ref{Sec:3}, we discuss the inference for multivariate versions of GLMMs. 
Section \ref{SubSec:3.1} presents two simulation studies. The appendices 
\ref{app:A0}, \ref{app:A1} 
and 
\ref{app:B} expose some technical details and involved calculations. 
Appendix \ref{app:C} presents a multivariate extension of the classical inference method based on a Laplace approximation for GLMMs.

% ---------------------------------
% ---------------------------------
\section{Extended One Dimensional Generalised Linear Mixed Models \label{Sec:2}}
% ---------------------------------
% ---------------------------------

This section will study a one-dimensional extension of standard GLMMs defined 
with random intercepts, and discuss an estimation technique based on conditional 
inference for those models. The GLMMs that we consider contain random 
components that are 
not necessarily Gaussian distributed. Moreover, they allow the conditional 
distributions to 
follow a general dispersion model, and therefore, they enlarge the class of 
standard 
GLMMs. We extend the models and inferential techniques described here to a 
multivariate context in Section \ref{Sec:3}.

% ---------------------------------
% ---------------------------------
\subsection{Generalised Linear Mixed Models with Simple Random 
Components \label{SubSec:2.0}}
% ---------------------------------

Consider the situation where we observe the responses of $n$ individuals or experimental units. Those responses are viewed as realisations of $n$ random variables taking values in $\Y\subseteq\R$, which we denote by  $Y_1,\ldots,Y_n$. 
Here  $\Y$ is typically $\R$, $\R_+$, a compact real interval or $\Z_+$
(corresponding to models defined using for example the Normal, Gamma, von Mises or the Poisson distributions).
Suppose that each individual belongs to one, and only one, of $q$ groups of individuals, referred as \emph{clusters}.
We assume that there exist $q$ independent unobservable random variables taking 
values in $\R$, say $B_1,\ldots,B_q$, termed the \emph{random components}, 
that will be associated to the clusters as described below. Denote the random 
vector $(B_1,\ldots,B_q)$ by $\bm{B}$.
According to the model, the responses $Y_1,\ldots,Y_n$ are conditionally independent given $\bm{B}$. Furthermore, for $i=1,\ldots,n$ and each $\bm{b}\in \R^q$, $Y_i$ is conditionally distributed according to a dispersion model  (see \citeauthor{Jorgensen1997},   \citeyear{Jorgensen1997} and \citeauthor{Cordeiro2021}, \citeyear{Cordeiro2021}, or equation (\ref{eq:dispDen}))
 given $\bm{B}$, with conditional expectation given by
\begin{align}\label{eqn3.2.01}
g\left (\E\left [Y_i \vert \bm{B}=\bm{b}\right ]\right ) = \bm{x}_i^T\bm{\beta}+\bm{z}_i^T\bm{b}\, , \,\,\, \mbox{ for all } \bm{b}\in \R^q \, .
\end{align}
Here $g$ is a given \emph{link function}, $\bm{x}_i$ is a vector of $k$ explanatory variables associated to the $i\tth$ individual and $\bm{\beta}\in \Omega\subseteq\R^k$ is a vector of coefficients,  referred as the \emph{fixed effects}.
Furthermore, $\bm{z}_{i}$ is a $q$-dimensional allocation vector associating the $i\tth$ individual to one of the $q$ clusters. 
The $j\tth$ entry of the vector $\bm{z}_{i}$ takes the value $1$ if the $i\tth$ individual belongs to the  $j\tth$ cluster and $0$ otherwise. Other forms of allocation vectors are possible, but we restrict to the particular form above to simplify the exposition of ideas.

It is convenient to introduce the following nomenclature and notation for the right side of (\ref{eqn3.2.01}).
The \emph{linear predictor} and the  \emph{conditional mean response} for the $i\tth$ individual ($i=1,\ldots , n$) are defined by $\eta_i = \eta_i \left (\bm{\beta}, \bm{b}\right ) \deffeq  \bm{x}_i^T\bm{\beta}+\bm{z}_i^T\bm{b}$  and $\mu_i = \mu_i \left (\bm{\beta}, \bm{b}\right ) \deffeq g^{-1}(\eta_i )$, respectively. The parameter space of the conditional means is denoted by $\U\subseteq\R$ and we write $\mu_i\in\U$. Additionally, denote the random vector of observations $\left ( Y_1, \ldots , Y_n \right )$ by $\BY$, and the vector of observed responses $\left ( y_1, \ldots , y_n \right )$ by $\By$.

The specification of the extended GLMM that we consider is completed by 
defining the distribution of the random components as follows.
 We assume that $B_1,\ldots, B_q$ are independent and identically distributed according to a distribution that is absolutely continuous with respect to the Lebesgue measure on $\R$, symmetric around zero, unimodal, and possesing finite moments up to the fourth-order. 
 Note that the random components have expectation zero due to the symmetry.
 Denote the density of this distribution by $\varphi(\cdot,\sigma^2)$, where $\sigma^2\in \V \deffeq \R_+$ is a parameter describing the dispersion of the distribution. Here a typical choice would be a normal or a regular absolute continuous one-dimensional elliptically contoured family of distributions and in this case  $\sigma^2$ would be the variance parameter.
\footnote{Here a one-dimensional elliptically contoured family of distributions is a location and scale family of distributions, with location and scale parameters $\mu$ and $\sigma$,  for which the characteristic functions $\phi$, satisfy the functional equation 
$\phi (t)=e^{i \mu t}\psi (-\tfrac{1}{2}t \sigma^2 t)$ for all $t\in\R$, for a given 
function  $\psi$.}

 Under the model defined above, the conditional distribution of the $i\tth$ 
 observation $Y_i$ given  $\bm{B}$ (for $i=1,\ldots,n$), has a density with 
 respect to a dominating measure $\nu$ (defined on the measurable space $(\Y, 
 \A)$),
%  , where $\A$ is a suitable $\sigma$-algebra), 
 taking the form of a dispersion model (see \citeauthor{Jorgensen1997},   \citeyear{Jorgensen1997} and \citeauthor{Cordeiro2021}, \citeyear{Cordeiro2021}). 
 Therefore, the refferred density takes the form
\begin{align}\label{eq:dispDen}
f(y_i \vert \bm{B}= \bm{b};\bm{\beta},\lambda)  
& = 
a(y_i; \lambda) \exp\left [ -\tfrac{1}{2\lambda} \, 
                                           d \left \{ y_i;g^{-1} ( \bm{x}_i^T\bm{\beta}+\bm{z}_i^T\bm{b}) \right \}\right ]
\\ \nonumber
& = 
a(y_i; \lambda) \exp\left \{-\tfrac{1}{2\lambda} \, d \left (y_i;\mu_i \right )\right \}, \quad \forall \,\, y_i \in\Y, \,\, \forall \,\, \bm{b}\in\R^q \, ,
\end{align} 
where $\bm{\beta}\in \Omega$ and $\lambda\in \Lambda = \R_+$.
The function $d: \mathcal{Y}\times \U \rightarrow \R_+$ is the \emph{unit 
deviance} and, by definition, satisfies that $d(\mu ,\mu)=0$ and $d(y,\mu)>0$ 
for all $(y,\mu)\in \mathcal{Y}\times \U$ such that $ y\neq \mu$. The function 
$a:\mathcal{Y}\times \R_+ \rightarrow \R_+$ is a given normalising function. 
We assume that the unit deviance is regular, that is, $d$ is twice continously 
differentiable in $\mathcal{Y}\times  \U$ and 
$\partial^2d(\mu;\mu)/\partial\mu^2>0$ for all $\mu\in\U$. The function 
$V:\U\rightarrow\R_+$ given by $V(\mu)=2/\{\partial^2 
d(\mu,\mu)/\partial\mu^2\}$  for all $\mu$ in $\U$ is termed the \emph{variance 
function} \citep{Cordeiro2021}. 
The conditional variance of $Y_i$ given the random components is $V(\mu_i)/\lambda$. 
The following families of distributions are examples of dispersion models: Normal, Gamma, inverse Gaussian, von Mises, Poisson, and Binomial families.

We formally define the extended GLMM described above as the family
\begin{align*}
 \PP = 
  \left \{
    P_{\Bbeta, \lambda, \sigma^2} :
    \Bbeta \in \Omega , \quad
    \lambda \in \Lambda = \R_+ , \quad
    \sigma^2 \in \V = \R_+
  \right \}
\end{align*}
of probability measures defined on the product measurable space $(\mathcal{Y}^n,\mathcal{A}^n)$ (where $\mathcal{A}^n$ is the related product $\sigma$-algebra) corresponding to the probability measures defining the extended GLMM described above. 
Let $\bm{\nu}$ be the product measure induced by $\nu$.
The density of the distributions in $\PP$, with respect to $\bm{\nu}$, are given by
\begin{align}\label{Eqn.2.2.1}
  p \left (\By ;  \bm{\beta},\lambda, \sigma^2 \right ) 
  \deffeq
   \frac{dP_{\Bbeta, \lambda, \sigma^2}(\By)}{d\bm{\nu}} 
    =
  \int_{\R^q} 
  \prod_{i=1}^n f(y_i\vert \bm{B} = \bm{b};\bm{\beta},\lambda)  
  \prod_{j=1}^q \varphi(b_j;\sigma^2) \,\, d \bm{b}  \, ,
\end{align}
%\begin{align}\label{Eqn.2.2.1}
%  p \left (\By ;  \bm{\beta},\lambda, \sigma^2 \right ) 
%  & \deffeq  
%   \frac{dP_{\Bbeta, \lambda, \sigma^2}(\By)}{d\bm{\nu}} 
%   \\ & =
%  \int_{\R^q}
%  \prod_{i=1}^n f(y_i\vert \bm{B} = \bm{b};\bm{\beta},\lambda)  
%  \prod_{j=1}^q \varphi(b_j;\sigma^2) \,\, db_1 \ldots db_q
%   \, ,
%\end{align}
for all $\By\in\Y^n$, $ \Bbeta \in \Omega$, $\lambda \in \Lambda$ and $\sigma^2 \in \V$.

We will use the following set of regularity conditions on the generalised linear mixed model $\PP$:
\begin{enumerate}[label=(\roman*)]
\item The matrices $\bm{X}$ and $\bm{Z}$ have full rank (\ie rank $k$ and $q$, respectively)
\item The link function is strictly monotone, invertible and  twice continuously differentiable with bounded first order derivative
\item The unit deviance, $d(y,\mu)$, is twice continuous differentiable with respect to $\mu$
\item  The  functions
$\tfrac{\partial}{\partial \bm{\beta}}d\left\{ \,\cdot\, 
;g^{-1}(\bm{x}_i^T\bm{\beta}+{\bm{z}}_i^T\bm{b})\right\}$.
and
$\tfrac{\partial}{\partial \bm{b}}d\left\{ \,\cdot\, 
;g^{-1}(\bm{x}_i^T\bm{\beta}+{\bm{z}}_i^T\bm{b})\right\}$
are  dominated by integrable functions (not necessarily the same dominating 
functions) for each $\bm{\beta}\in \R^k$ and $\bm{b}\in \R^q$.
\end{enumerate}
These mild regularity conditions turn out to be minimal requirements for the inference theory that we construct.

Let $\bm{y} = (y_1, \ldots y_n)$ be a realisation of the random vector $\bm{Y} 
\deffeq (Y_1,\ldots,Y_n)$ of responses. 
Under the model $\PP$, the likelihood function for the parameters $\bm{\beta},\lambda$ and $\sigma^2$, based on $\bm{y}$, is 
\begin{align}\label{eq:margLik}
L(\bm{\beta},\lambda, \sigma^2; \bm{y}) 
= 
p \left (\By ;  \bm{\beta},\lambda, \sigma^2 \right ) 
% \int_{\R^q}\prod_{i=1}^n f(y_i\vert \bm{b};\bm{\beta},\lambda)  
%                   \prod_{j=1}^q \varphi(b_j;\sigma^2) \,\, db_1 \ldots db_q
\, . 
\end{align}
Usually, the integral in the right side of \eqref{Eqn.2.2.1} involved in the 
calculation of the likelihood function in  \eqref{eq:margLik}, cannot be evaluated 
in closed form. In Section \ref{SubSec:2.1}, we introduce an inference method that 
includes predictions of values of the random components, $B_1,\ldots, B_q$, and 
avoids the integration.
This inferential procedure will be justified using asymptotic arguments in Section \ref{SubSec:2.2}.

%%%
%%%
%%%

We introduce below two families of probability measures related to $\PP$, which will be convenient for presenting and discussing the conditional inference for the GLMMs under discussion.
% introduced above. 
First, consider a statistical model, $\overline\PP$, constructed on $ \Y^n \times \R^q $,
collecting the joint distributions of the $n$ responses and the $q$ random components. 
This model, called the \emph{joint-model}, 
represents the hypothetical situation in which the random components would be observable.
We will use the join-model to introduce and motivate the inferential techniques we propose.

%%%
%%%
%%%

It is convenient to introduce also the following family of probability measures on $(\mathcal{Y}^n,\mathcal{A}^n)$,
obtained by collecting the distributions constructed with the realisable values of the random components $B_1,\ldots, B_q$, in the following way
\begin{align} \label{eqn:2.1.3}
\PP^* =
 \left \{
 \!\!
 \begin{array}{ll}
  P_{\bm{\beta},\bm{b},\lambda}^*  : 
   \frac{d P^*_{\bm{\beta},\bm{b},\lambda}}{d \bm{\nu}}(\bm{y}) & = 
   f^*\left (\bm{y}; \bm{\beta},\bm{b},\lambda \right )  
   \\ &
   \mbox{ for all }
   \bm{y} \in \mathcal{Y}^n, 
  \bm{\beta}\in \Omega,\bm{b} \in \R^q, \lambda \in \Lambda 
  \end{array}
  \!\!
   \right \}
  .
\end{align}
The density of the probability measure referred above is given by
\begin{align*}% \label{eqn:2.1.4}
f^*\left (\bm{y}; \bm{\beta},\bm{b},\lambda \right )  \deffeq 
 \prod_{i=1}^n  f(y_i\vert \bm{B}=\bm{b};\bm{\beta},\lambda) =
 % \\ & = 
 \prod_{i=1}^n
   a(y_i; \lambda) \exp\left \{-\tfrac{1}{2\lambda} \, d \left [ y_i;  \mu_i \left (\bm{\beta}, \bm{b}\right )\right ] \right \} \, ,
\end{align*}
for all $\By\in\Y^n$, $ \Bbeta \in \Omega$, $\lambda \in \Lambda$ and $\bm{b} \in \R^q $.
We call the family $\mathcal{P}^*$ the \emph{conditional model}. 
This family will be used for defining inference functions, and establishing the 
basic properties of the inference procedures we will propose.

% ---------------------------------
% ---------------------------------
\subsection{Conditional Inference for Models with a Single Random Component \label{SubSec:2.1}}
% ---------------------------------

Under the joint model $\overline\PP$, the log-likelihood function for estimating $\bm{\beta}$, $\lambda$ and  $\sigma^2$ 
% (\ie the logarithm of the joint density of $\bm{Y}$ and $\bm{B}$) 
based on realisations  $\bm{y}$ and $\bm{b}$ of $\bm{Y}$ and $\bm{B}$, respectively, is
\begin{align}\label{eqn:2.1.1}
 % \log \left\{ L({\bm{\beta}},{{\lambda}}, \sigma^2; \bm{y},\bm{b})\right\} = 
 l \left (\bm{\beta},{\lambda}, \sigma^2; \bm{y},\bm{b} \right ) \deffeq
  \sum_{i=1}^n \log f(y_i\vert \bm{B}=\bm{b};\bm{\beta},\lambda) + 
  \sum_{j=1}^q \log \varphi(b_j;\sigma^2 ) \, .
\end{align}
From this perspective, $\bm{b}$ is a S-sufficient statistic with respect to 
$\sigma^2$ (since the term of the likelihood function that contains $\sigma^2$ 
depends only on  $\bm{b}$ and not on $\bm{y}$),
and S-ancillary with respect to  $\bm{\beta}$ and ${\lambda}$ (since the term of the likelihood function that contains $\bm{\beta}$ and ${\lambda}$ involves $\bm{b}$ only conditionally). 
See \citeauthor{BarndorffNielsen2014} (\citeyear{BarndorffNielsen2014},  page 50) or
\citeauthor{Jorgensen2012} (\citeyear{Jorgensen2012},  Section 3.2) 
for formal definitions.

The decomposition of the likelihood function of the joint model $\overline\PP$, defined in (\ref{eqn:2.1.1}), motivates that the inference on $\sigma^2$ should be performed using the term
\begin{align*}
\sum_{j=1}^q \log \varphi({{b}}_j;\sigma^2 ) \, ,
\end{align*}
corresponding to base the inference on $\sigma^2$ on a sufficient statistic.
 Following the same line, the inference on ${\bm{\beta}}$ and ${{\lambda}}$ should be performed only using the term
\begin{align}\label{eqn:2.1.2}
 \sum_{i=1}^n \log f(y_i\vert \bm{B}=\bm{b};\bm{\beta},\lambda) ,
\end{align}
which corresponds to perform conditional likelihood-based inference given an ancillary statistic. 
Therefore, we propose to estimate $\bm{\beta}$ and ${\lambda}$ by inserting a reasonable prediction of $\bm{b}$, say $\tilde{\bm{b}}$ as defined below, into \eqref{eqn:2.1.2} and maximising for $\bm{\beta}$ and ${\lambda}$.
We argue in Section \ref{SubSec:2.2} that the procedure informally defined here yields sensible estimates.

We turn now to the problem of predicting $\bm{b}$. 
Under the joint model 
$\overline\PP$, it is natural to predict $\bm{b}$ by maximising 
$ l \left (\bm{\beta},{\lambda}, \sigma^2; \bm{y},\bm{b} \right )$ given in \eqref{eqn:2.1.1},
\ie by 
\begin{align}\label{eq:bhat1}
\hat{\bm{b}}({\bm{\beta}},{{\lambda}},\sigma^2;\bm{y}) = 
\argmax_{b_1,\ldots,b_q} 
\!
\left \{ 
\!
\sum_{i=1}^n \log f(y_i \vert \bm{B}=\bm{b},\bm{\beta},\lambda) + \sum_{j=1}^q\log \varphi(b_j;\sigma^2 )
\right \} .
\end{align}
However, it is convenient, as we will demonstrate in Section \ref{SubSec:2.2}, to use the following approximation to $\hat{\bm{b}}$,
\begin{align}\label{eq:btilde}
 \tilde{\bm{b}}({\bm{\beta}},{{\lambda}};\bm{y})  \deffeq 
 \Pi_{\B }  
   \left (
    \argmax_{b_1,\ldots,b_q} 
    \sum_{i=1}^n \log f(y_i\vert \bm{B}=\bm{b},\bm{\beta},\lambda) 
    \right ),
\end{align}
where $\B  \deffeq \{ \Bb\in\R^q : \tfrac{1}{q}\sum_{j=1}^q b_j = 0 \}$ is the 
subspace of  the vectors in $\R^q$ with mean zero, and  $ \Pi_{\B } : \R^q 
\rightarrow \B $ is the projection function given by 
$ \Pi_{\B } (\bm{y}) \deffeq \bm{y} - 1/q \sum_{j=1}^q y_j$.
Note, that $ \tilde{\bm{b}}$ is an approximation of $ \hat{\bm{b}}$, because the 
last term of the right side of \eqref{eq:bhat1} is maximised by setting $\bm{b}$ 
equal to zero. The approximation follows from the continuity of the function 
$\varphi(\cdot;\sigma^2 )$, which has a unique mode at zero,  and because 
$\tilde{\bm{b}}$ is in $\B $.

% ---------------------------------
% ---------------------------------
\subsection{Asymptotic Properties of the Conditional Inference Method  \label{SubSec:2.2}}
% ---------------------------------

In this section, we formulate the inferential techniques  presented in Section 
\ref{SubSec:2.1} using the theory of inference functions (\cite{Jorgensen2012} 
and \cite{BarndorffNielsen2014}). 
We show that the estimated value of $\bm{\beta}$ and the predicted values of $\bm{b}$ are asymptotically Gaussian distributed when the variance, $\sigma^2$,  of the random component is small.

We consider below the  inference functions
\begin{align*}
	\psi^*_{\bm{\beta}}:  \Omega \times  \R^q  \times \mathcal{Y} \rightarrow 
	\R^k
	\mbox{ and }
	\psi^*_{\bm{b}}:    \Omega \times  \R^q  \times \mathcal{Y}  \rightarrow 
	\R^{q},
\end{align*}
%\begin{align*}
%	&\psi^*_{\bm{\beta}}:  \Omega \times  \B  \times \mathcal{Y} \rightarrow 
%\R^p\\
%	&\psi^*_{\bm{b}^*}:    \Omega \times  \B  \times\times \mathcal{Y}  
%\rightarrow \R^{q},
%\end{align*}
%%
%%
which are equivalent to the score functions for estimating $\Bbeta$ and $\Bb$, 
under  $\PP^*$,  with $\lambda$ treated as a nuisance parameter. 
The inference functions  $\psi^*_{\bm{\beta}}$ and $\psi^*_{\bm{b}}$ 
referred above are defined by 
\begin{align} \label{eq:inf1}
	\psi^*_{\bm{\beta}}(\bm{\beta},{\Bb};\bm{y}) & =
	\sum_{i=1}^n 
	\tfrac{\partial}{\partial \bm{\beta}}
	d\left ( y_i;g^{-1}(\bm{x}_i^T\bm{\beta}+ \tilde{\bm{z}}_i^T \Bb)\right )
	=
	\sum_{i=1}^n 
	\bm{x}_i
	\frac{\tfrac{\partial}{\partial {\mu}_i}
		d(y_i;{\mu}_i)}
	{g'({\mu}_i)}, 
	\\ \label{eq:inf2}
	\psi^*_{{\Bb}}(\bm{\beta}, {\Bb};\bm{y}) &=
	\sum_{i=1}^n 
	\tfrac{\partial}{\partial \Bb}
	d\left ( y_i;g^{-1}(\bm{x}_i^T\bm{\beta}+ \tilde{\bm{z}}_i^T \Bb)\right )
	=
	\sum_{i=1}^n 
	\bm{z}_i
	\frac{\tfrac{\partial}{\partial {\mu}_i}
		d(y_i;{\mu}_i)}
	{g'({\mu}_i)} 
	\, .
\end{align}
Note that the score functions for estimating $\Bbeta$ and $\Bb$ are given by  
$\psi^*_{\bm{\beta}}$ and $\psi^*_{\bm{b}}$  multiplied by 
$-\tfrac{1}{2\lambda}$. 
However, since $\lambda$ is a positive number the solution to the score equations 
for  $\Bbeta$ and $\Bb$ are exactly the roots of $\psi^*_{\bm{\beta}}$ and 
$\psi^*_{\bm{b}}$; in this sense they are equivalent.
The inference function $\psi^*:  \Omega \times  \R^q  \times \mathcal{Y} 
\rightarrow \R^{k+q}$ given by
\begin{align*}
	\psi^*(\Bbeta,\Bb;\bm{y})=\left\{\left[\psi^*_{\bm{\beta}}(\Bbeta,\Bb;\bm{y})\right]^T,
	\left[\psi^*_{\bm{b}}(\Bbeta,\Bb;\bm{y})\right]^T\right\}^T
	 ,
\end{align*}
will be used for estimating $\Bbeta$ and predicting $\bm{b}$.
We denote the sequences of  roots of the inference functions  
$\psi^*_{\bm{\beta}}$ and $\psi^*_{\bm{b}}$ by  
$\{\widehat{\bm{\beta}}_n\}_{n\in \mathbb{N}}$ and  
$\{\widehat{\bm{b}}_n\}_{n\in \mathbb{N}}$,  
respectively, obtained when the number of observations, $n$, increases.

According to the classic theory of inference functions (see 
\citeauthor{Jorgensen2012}, \citeyear{Jorgensen2012}, Chapter 4), the 
estimating functions $\psi^*_{\bm{\beta}}$ and $\psi^*_{\bm{b}}$ yield 
consistent estimates under $\PP^*$. Moreover, the estimates of $\Bbeta$ and 
$\Bb$ are conditionally asymptotically normally distributed (see the details in 
appendix \ref{app:A1}). 
However, our primary interest is on estimating $\Bbeta$ under the extended 
generalised linear mixed model $\PP$.
For this purpose, we define below the inference function $\psi_{\bm{\beta}}:  
\Omega  \times \mathcal{Y} \rightarrow \R^k$ given by 
\begin{align} \label{eq:infFunPhi}
	\psi_{\bm{\beta}}(\bm{\beta};\bm{y}) \deffeq
	\psi^*_{\bm{\beta}}(\bm{\beta},\widehat{\Bb};\bm{y}), \mbox{ for all } 
	\Bbeta\in\Omega \mbox{ and all } \By \in \Y\, ,
\end{align}
where $\widehat{\Bb}$ is obtained from the joint solution, $(\widehat{\bm 
	{\beta}}, \widehat{\Bb})$, of the estimating equation 
$\psi^*_{\bm{\beta}}(\bm{\beta},\Bb;\bm{y})=0$ and 
$\psi^*_{\Bb}(\bm{\beta}, \Bb;\bm{y})=0$.
The theorem below shows that, under the assumed mild regularity conditions, the 
root of $\psi_{\bm{\beta}}$ are consistent  and asymptotically Gaussian 
distributed when the variance of the random components converges to zero.

\begin{theorem} \label{theor:AsymptNormality}
	Under the regularity conditions $i$-$iv$, the sequences  
	$\{\hat{\bm{\beta}}_n\}_{n\in \mathbb{N}}$ and  
	$\{\widehat{\bm{b}}_n\}_{n\in \mathbb{N}}$ are 
	consistent (in probability) under $\PP^*$.
	Moreover, $\{\hat{\bm{\beta}}_n\}_{n\in \mathbb{N}}$ is consistent (in 
	probability) under $\PP$. Both sequences are asymptotically Gaussian 
	distributed,
	when $n\to \infty$ and  $\sigma^2 \downarrow 0$.
\end{theorem}
\begin{proof}
	See Lemma \ref{appLem:asympNorm} for the consistency of  
	$\{\widehat{\bm{\beta}}_n\}_{n\in \mathbb{N}}$ and  
	$\{\widehat{\bm{b}}_n\}_{n\in \mathbb{N}}$ under $\PP^*$.
	See Lemma \ref{Lemma:ConsistP} for the consistency in probability of 
	$\{\hat{\bm{\beta}}_n\}_{n\in \mathbb{N}}$ under $\PP$ and  Theorem 
	\ref{prop:AsymConvApp} in Appendix \ref{App:proof} for the asymptotic 
	normality when the variance of the random components is sufficiently small.
\end{proof}

The parametrisation of the family $\PP^*$ defined above is not identifiable.
Note, that a natural parametrisation of $\PP^*$ using 
the triplet 
$(\Bbeta, \Bb, \lambda)\in \Omega\times\R^q\times\Lambda$
is not identifiable.
Indeed, according to the Lemma \ref{LemmaA01} proved in the appendix \ref{app:A0}, 
for any $i\in\{1, \ldots ,n\}$ and any choice of $\Bbeta$ ,  $\Bb$ and $\delta>0$
there exists a $\bm{\beta}_\delta \in\Omega$ such that 
$\eta_i\left ( \bm{\beta} , \bm{b} \right ) = \eta_i\left ( \bm{\beta}_\delta , \bm{b} - \delta \right )$. 
A convenient solution to this issue is to introduce a constraint and require that $\Bb$ takes values in $\B $ (\ie the sub-space of $\R^q$ of vectors with mean zero), which yields an identifiable parametrisation of $\PP^*$. We adopt this parametrisation and re-write here \eqref{eqn:2.1.3} in the form
\begin{align*} 
\PP^* =
 \left \{
 \!\!
 \begin{array}{ll}
  P_{\bm{\beta}^*,\bm{b^*},\lambda}^*  : 
   \frac{d P^*_{\bm{\beta}^*,\bm{b^*},\lambda}}{d \bm{\nu}}(\bm{y}) & = 
   f^*\left (\bm{y}; \bm{\beta}^*,\bm{b^*},\lambda \right )  
   \\ &
   \mbox{ for all }
   \bm{y} \in \mathcal{Y}^n, 
  \bm{\beta}^*\in \Omega^*,\bm{b^*} \in \B , \lambda \in \Lambda 
  \end{array}
  \!\!
   \right \}
  ,
\end{align*}
so the mapping from $\Omega\times\B\times\Lambda $ to $\PP^*$  given by
$(\Bbeta^*, \Bb^*, \lambda) \mapsto P_{\Bbeta^*,\Bb^*, \lambda}$
is a bijection.

The sequences of estimates 	$\{\widehat{\bm{\beta}}^*_n\}_{n\in \mathbb{N}}$ 
and  
$\{\widehat{\bm{b}}^*_n\}_{n\in \mathbb{N}}$ obtained as roots to the inference 
functions defined as above but with the new identifiable parametrisation, yields 
the same maximum likelihood values as a consequence of Lemma 
\ref{LemmaA01} proved in the appendix \ref{app:A0}. By the law of large 
numbers and Lemma~\ref{appLem:asympNorm}, 
$(\widehat{\bm{\beta}}^*_n,\widehat{\bm{b}}^*_n)$ converges to 
$(\widehat{\bm{\beta}}_n,\widehat{\bm{b}}_n)$ in probability under 
$P_{\Bbeta,\Bb, \lambda}^*$ for 
$q$ and $n$ converging to 
infinity.

In Section~\ref{SubSec:3.2}, we study the distribution of $\bm{\hat{\beta}}$ in a 
simulated example, where we assume that the random components follow a 
Gaussian distribution.

% ---------------------------------
% ---------------------------------
\subsection{A Simple Algorithm for Conditional Inference \label{SubSec:2.3}}
% ---------------------------------

% In this section, we describe an algorithm based on the inference method described above. 
The following algorithm implements the inference method described above. 
The algorithm starts by setting the initial values $\bm{\beta}^{(0)}$ and 
$\lambda^{(0)}$ for the parameters $\bm{\beta}$ and $\lambda$. We used the 
estimated values of the corresponding parameters of a generalised linear model 
defined with the same distribution and link function as in the extended GLMM in 
study, and with the linear predictor given by the fixed effects of the extended 
GLMM in 
discussion. The algorithm repeats the following two steps, starting with $m=0$, 
until convergence:
\begin{enumerate}
	\item Let $\bm{\beta}^{(m)}$ and $\lambda^{(m)}$ be the current estimates 
	of the parameters $\bm{\beta}$ and $\lambda$. Set
	\begin{align*}
		\bm{b}^{(m+1)} =               \argmax_{b_1,\ldots,b_q} 
		\sum_{i=1}^n \log f(y_i\vert 
		\bm{B}=\bm{b},\bm{\beta}^{(m)},\lambda^{(m)}), 
	\end{align*}
	and
	\begin{align*}
	\bm{b}^*_{(m+1)}= 
	\tilde{\bm{b}}(\bm{\beta}^{(m)},\lambda^{(m)};\bm{y})
	= 
         \Pi_{\B }  
         \left ( \bm{b}^{(m+1)}
          \right ),
	\end{align*}
with $\tilde{\bm{b}}$ is defined as in \eqref{eq:btilde}.
	\item  Given the latest predicted values of the random components denoted 
	$\bm{b}^*_{(m+1)}$, $\bm{\beta}^{(m+1)}$ and $\lambda^{(m+1)}$ are 
	estimated by maximising \\$\prod_{i=1}^n f^*(y_i;\bm{\beta}, 
	\bm{b}^*_{(m+1)}, 
	\lambda)$ with respect to $\bm{\beta}$ and $\lambda$.
\end{enumerate}

\noindent
After convergence has been obtained, we estimate the variance, finding the value of $\sigma^2$ that maximises the integral
\begin{align}
	\int_{\B} 
	g(\hat{\bm{b}};\bm{b},\bm{\Sigma}_{\hat{\bm{b}}})
	\prod_{j=1}^q \varphi({b}_j;\sigma^2)d\bm{b},\label{eq:varInt}
	\, ,
\end{align}
where $\hat{\bm{b}}$ denotes the value of 
$\bm{b}^{(m+1)}$ in the last round of the algorithm. 
Here, 
$g(\cdot;\bm{b},\bm{\Sigma}_{\hat{\bm{b}}})$ denotes the density of the 
predicted values from the final iteration, $\hat{\bm{b}}$, with expectation 
$\bm{b}$ and covariance $\bm{\Sigma}_{\hat{\bm{b}}}$. In the case where 
$\sigma^2$ is small enough and $n$ is large enough, this density is close to the 
multivariate Gaussian 
density, see Theorem~\ref{theor:AsymptNormality} for details. 
%The covariance 
%matrix, $\bm{\Sigma}_{\hat{\bm{b}}}$, can be replaced with 
%$J_{\hat{\bm{b}}}^{-1}(\hat{\bm{\beta}},\hat{\bm{b}})$ calculated in the 
%final 
%iteration to obtain a profile likelihood.
 In Appendix~\ref{app:B}, calculations of the above integral are given in the case where $g$ and $\varphi$ are densities of Gaussian distributions.
% ---------------------------------
% ---------------------------------
\subsection{Conditional Inference for Models with Complex Random 
Components  \label{SubSec:2.4}}
% ---------------------------------

This section extends the methods introduced in section \ref{SubSec:2.2} to a 
context with complex random components. We first consider non-nested 
random components, and then we study a scenario where the random components 
are nested or a combination of the two cases. 

When the random components are not nested, the values  of the random 
components are easily predicted using the already described method. To simplify 
the notation, consider a one dimensional extended GLMM with two vectors of 
non-nested 
random components (each corresponding to a clustering of the observations), say 
$\bm{B}_1$ and $\bm{B}_2$ with length $q_1$ 
and $q_2$, respectively. 
 We assume that $Y_1,\ldots,Y_n$ are conditional independent random variables given  $\bm{B}_1$ and $\bm{B}_2$, and conditionally distributed according to a dispersion model, with conditional density 
 $f( \, \cdot \, \vert \bm{B}_1=\bm{b}_1,\bm{B}_2=\bm{b}_2,\bm{\beta},\lambda)$,
 where $f$ is defined in  \eqref{eq:dispDen}.

%\begin{align*}
%Y_i\mid \bm{B}_1=\bm{b}_1, \bm{B}_2=\bm{b}_2 \sim \text{D}(\mu_{i},\lambda) \quad  \forall \bm{b}_1\in \R^{q_1},\bm{b}_2\in \R^{q_2} \text{ and } i=1,\ldots,n,
%\end{align*}
%with $g(\mu_i)=\eta_i=\bm{x}_i^T\bm{\beta}+\bm{z}_{i1}^T\bm{b}_1 + \bm{z}_{i2}^T\bm{b}_2$, and $\bm{x}_i$ being a vector of explanatory variables and $\bm{z}_{i1}$, $\bm{z}_{i2}$ location vector for the two random components.

Recall, that values of the random components were predicted using 
Equation~\eqref{eq:btilde}, which is equivalent to solving the inference functions 
in \eqref{eq:inf1} and \eqref{eq:inf2}. This equation can easily be adapted to the 
situation with multiple non-nested random components.  To do so, we replace 
$\B$ by 
$
\widetilde\B  \deffeq 
\left\{ (\Bb_1, \Bb_2)\in\R^{q_1 + q_2} : 
\Bb_1\in \B (\R^{q_1}) \mbox{ and } \Bb_2\in \B(\R^{q_2}) \right \}
$
(where $\B(\R^{q})$ is the space of vectors of $\R^{q}$ with mean zero)
and define 
\begin{align*}% \label{eq:btilde}
 \tilde{\bm{b}}({\bm{\beta}},{{\lambda}};\bm{y})  \deffeq 
 \Pi_{\widetilde\B }  
   \left [
   \argmax_{(\bm{b}_1,\bm{b}_2)\in \R^{q_1+q_2}} \sum_{i=1}^n \log 
   f(y_i\vert \bm{B}_1=\bm{b}_1,\bm{B}_2=\bm{b}_2;\bm{\beta},\lambda)
    \right ].
\end{align*}

We turn now to the case of two nested vectors of random components 
$\bm{B}_1$ and 
$\bm{B}_2$, where $\bm{B}_1$ is nested in $\bm{B}_2$, that is, the clusters 
corresponding to the entries in $\bm{B}_2$ groups multiple clusters associated 
with $\bm{B}_1$. Therefore, the variation 
in $\bm{B}_1$ should be interpreted as the remaining variation not explained by 
$\bm{B}_2$.  In this case, we estimate the model including only the random 
component $\bm{B}_1$. After predicting (temporary) values for $\bm{B}_1$ 
denoted by $\bar{\bm{b}}_1$, we predict the final values of $\bm{b}_2$ by
\begin{align*}
\hat{\bm{b}}_2=(\bm{Z}_2^T\bm{Z}_2)^{-1}\bm{Z}_2^T\bar{\bm{b}}_1,
\end{align*}
where $\bm{Z}_2$ a $q_1\times q_2$ dimensional matrix with the $(i,j)$'th 
entry equal to one if the cluster corresponding to the $i\tth$ entry of $\bm{B}_1$ 
is contained in the $j\tth$ cluster associated with the $j\tth$ entry of $\bm{B}_2$,
 and zero otherwise. 
Next, the predicted values of $\bm{b}_1$ is updated to the final values by
\begin{align*}
\hat{\bm{b}}_1=\bar{\bm{b}}_1-\bm{Z}_2	\hat{\bm{b}}_2.
\end{align*}
These methods can easily be generalised to the multivariate case by using the 
approach described in Section~\ref{Sec:3}.

% ---------------------------------
% ---------------------------------
\section{Multivariate Models \label{Sec:3}}
% ---------------------------------
% ---------------------------------

 In this section, we extend the methods described so far in one dimension to a multivariate context.
Consider $d$ response vectors simultaneously observed, each of them following an  GLMM described in Section~\ref{Sec:2}. Here the $d$ responses 
might follow different dispersion models, use different link functions, but the $d$ 
marginal extended GLMMs must have a common  random component  with the 
same clusters for each of the response vectors.
The inference method presented in the Sections~\ref{SubSec:2.1} - 
\ref{SubSec:2.4} yields predicted values of the random components directly as an 
additional product of the estimation process. 
% Here, instead of estimating the 
% parameter $\sigma^2$ in each dimension, we estimate a $d\times d$-dimensional 
% covariance matrix giving the covariance of some locally structure affecting all 
% dimensions. 

% ---------------------------------
% ---------------------------------
\subsection{Basic Setup \label{SubSec:3.1}}
% ---------------------------------
% ---------------------------------

We introduce the following notation required for formally defining the 
multivariate model we have in mind.
Let $\bm{Y}=\{\bm{Y}_1,\ldots,\bm{Y}_d\}$ be a $n\times d$ dimensional 
response variable matrix, and 
$\bm{B}=\{\bm{B}_{(1)},\ldots,\bm{B}_{(d)}\}=\{\bm{B}_1,\ldots,\bm{B}_q\}^T
 $ a $q\times d$ dimensional matrix of random components. Each column of 
$\bm{Y}$ corresponds to $n$ response variables in a univariate model. We 
assume, that 
the rows of $\bm{B}$ are independent and identical distributed according to a 
multivariate distribution which is absolute continuous with respect to the 
Lebesgue 
measure,  symmetric around the vector of zeros, unimodal, and with finite 
moments up to fourth order.
We will let $\bm{\Sigma}$ denote a covariance matrix of the distribution and $\varphi(\cdot,\bm{\Sigma})$ the density. Often, this distribution will be assumed to be multivariate Gaussian with expectation zero and covariance matrix given by $\bm{\Sigma}$.

For $i=1,\ldots,n$ and $j=1,\ldots,d$, we assume that  $Y_{ij}$ is conditional distributed according to a dispersion model given $\bm{B}_{(j)}=\bm{b}_{(j)} $.
That is, $Y_{ij}\vert \bm{B}_{(j)}=\bm{b}_{(j)} \sim 
\text{D}(\mu_{ij},\lambda_j)$ for $i=1,\ldots,n$ and $j=1,\ldots,d$, where 
$D(\mu;\lambda)$ 
denotes the dispersion model distribution with expectation $\mu$ and dispersion 
$\lambda$. The conditional expectation, 
$\mu_{ij}$, is connected to the linear predictor, $\eta_{ij}$, through the known 
link function denoted $g_j$, that is, 
$g_j(\mu_{ij})=\eta_{ij}=\bm{x}_{ij}^T\bm{\beta}_j+\bm{z}_{i}^T 
\bm{b}_{(j)}$, where $\bm{x}_{ij}$ and $\bm{z}_{i}$ denote the vector of 
explanatory variables  and a location vector, respectively. Notice, that like in the 
one dimensional model, $\bm{z}_{i}$ has one entry equal to one and the 
remaining entries are equal to zero. Thus,  $\bm{z}_{i}$ has a one in the entry 
corresponding to the cluster that the $i$th individual 
belongs to. 
The conditional density of $Y_{ij}$ given $\bm{B}_{(j)}$ is denoted by $f_j$. 

We assume,  that $Y_{ij}$ and $Y_{i'j}$ are conditionally independent given 
$\bm{B}_{(j)}=\bm{b}_{(j)} $ for $i\neq i'$ ($i,i'=1,\ldots,n$).  Moreover, the 
structure of the model implies that $Y_{ij}$ and $Y_{i'j'}$ are conditionally 
independent given 
$\bm{B}_{(j)}$ and $\bm{B}_{(j')}$ for all $i,i'=1,\ldots,n$ and $j,j'=1\ldots,d$ 
such that $j\neq j'$.

% ---------------------------------
% ---------------------------------
\subsection{Simulation Studies \label{SubSec:3.2}}
% ---------------------------------

In this section, we present results of two simulation studies illustrating basic 
properties of the proposed estimation procedure. Moreover, we compare the 
behaviour of the proposed estimates with two other inference methods: the 
multivariate Laplace approximation suggested by \cite{Breslow1993} (see 
Appendix~\ref{app:C} for details) and a 
Hermite 
quadrature estimation procedure. 
Two simulation studies are presented to study the distribution of the estimates 
when the entries in the covariance matrix are varied, and the bias of the estimated 
parameters when we increase the numbers of clusters of 
the random component (and thereby the number of observations).
In both simulation studies, we simulate a two dimensional generalised linear 
mixed model, where  $Y_{ij}$ for $i=1,\ldots,n$ and $j=1,2$ denotes the 
response 
variables. We follow the notation introduced above and let $\bm{B}_{(1)}$ and 
$\bm{B}_{(2)}$ denote $q$-dimensional random vectors representing the 
random components in the model. We assume that $Y_{11},\ldots,Y_{n2}$ are 
conditionally independent given $\bm{B}_{(1)}$ and $\bm{B}_{(2)}$. 
Moreover, we assume that given $\bm{B}_{(1)}$ and $\bm{B}_{(2)}$, 
$Y_{i1}$ and $Y_{i2}$ are conditionally distributed according to a Gaussian and 
a Poisson distribution, respectively, with conditional expectations given by
\begin{align*}
	\E[Y_{i1}\vert 
	\bm{B}_{(1)}&=\bm{b}]=\bm{x}_{i1}^T\bm{\beta}+\bm{z}_{i}^T\bm{b} 
	\quad \text{for } i=1,\ldots,n,\\
	\E[Y_{i2}\vert \bm{B}_{(2)}&=\bm{b}]=\exp \big ( 
	\bm{x}_{i2}^T\bm{\beta}+\bm{z}_{i}^T\bm{b}\big) \quad \text{for } 
	i=1,\ldots,n,
\end{align*}
%where $\bm{\beta}=(\beta_1,\beta_2)=(1.9042703, 0.2139833)$. 
where $\bm{\beta}=(\beta_1,\beta_2)=(1.90, 0.21)$. 
The Gaussian 
conditional distribution is assumed to have a variance of $0.5$ which is not varied 
in the simulations.

We assume that $\bm{B}^T=(\bm{B}_{(1)}^T,\bm{B}_{(2)}^T)$ is Gaussian 
distributed with expectation zero and covariance structure given by
\begin{align*}
	\Cov({B}_{(1)}^l,{B}_{(2)}^l)&= \bm{\Sigma} \quad \text{for } \, l=1,...,q, 
	\\
	\Cov({B}_{(1)}^l,{B}_{(2)}^{k})&=\,0\, \quad \text{for } \, l,k=1,...,q \text{ 
	such that } l\neq k,
\end{align*} 
where ${B}_{(j)}^l$ denotes the $l\tth$ entry in $\bm{B}_{(j)}$ for $j=1,2$, 
and 
\begin{align}
	\bm{\Sigma} = \text{const}  \begin{pmatrix}
		0.28 & 0.09 \\
		0.09 & 0.12
	\end{pmatrix} \label{eq:simCov},
\end{align}
with the constant depending on the simulation study. That is,
\begin{align*}
	\bm{B} \sim N_{2q}(\bm{0},\bm{\Sigma}\otimes\bm{I}_n),
\end{align*}
where  $\bm{I}_m$ denotes a $m$-dimensional identity matrix.

In the first simulation study, we simulate the above described model for three 
different covariance matrices, corresponding to three different values of the 
constant in \eqref{eq:simCov}. In that way, we can examine the sensitivity in the 
normality of the estimates to an increase in the variance. 
Theorem~\ref{theor:AsymptNormality} states that under some regularity 
conditions, 
the estimated values of $\bm{\beta}$ should be Gaussian distributed when the 
variance of the random components goes to zero. That is, the lower the constant in 
\eqref{eq:simCov} is, the closer is the distribution of $\bm{\beta}$ to a Gaussian 
distribution. In this simulation study, we used the following constants:  $c_1=1$, 
$c_2=50$ and $c_3=100$. 
In each of the three simulation studies we simulate $500$ datasets and estimate the above described model for each simulation. The results are presented in Figure~\ref{fig:qqplotSim}.

In the second simulation study, we fix the covariance matrix of the random 
components to $\bm{\Sigma}$ defined in \eqref{eq:simCov} with the constant 
set to one. In this study, we vary the lengths of $\bm{B}_{(1)}$ and 
$\bm{B}_{(2)}$ between the values $10$, $50$ and $100$, whereas the 
lengths was fixed to $60$ in the above described simulation study.
For each value of $q$ (the length of each vector of random components), we 
simulate the model $500$ times and estimate the bias and standard errors of the 
parameters. 

\begin{figure}[htbp!]
	\centering
\includegraphics[width=0.8\textwidth]{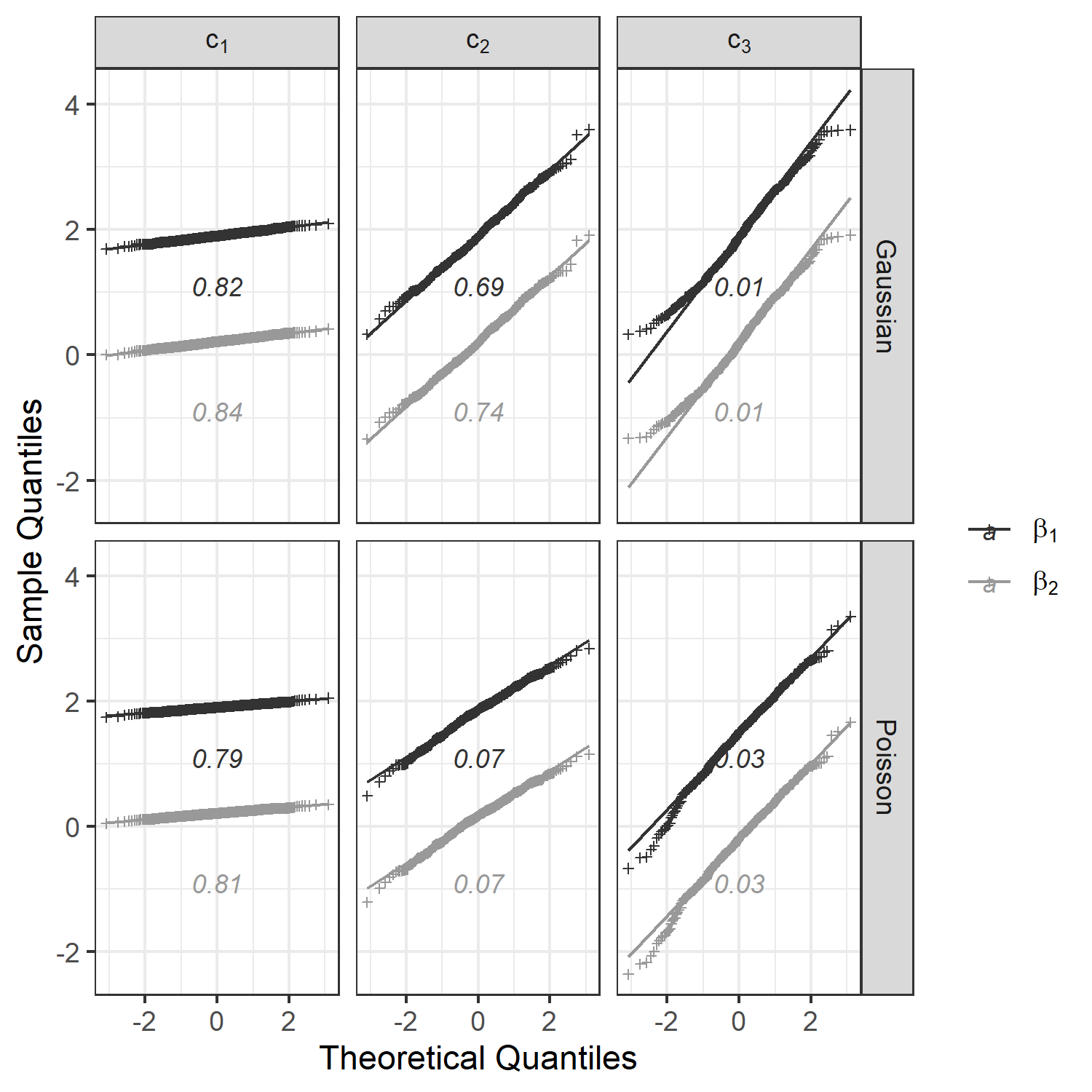}
\caption{QQ-plot of the theoretical Gaussian quantiles versus the sample 
quantiles of the estimated values of $\beta_1$ and $\beta_2$ in the described multivariate 
generalised linear mixed model for different sizes of $\bm{\Sigma}$. The numbers in the plots are the resulting 
p-values from Shapiro Wilk tests for normality. \label{fig:qqplotSim}}
\end{figure}

\begin{figure}[htbp!]
	\centering
\includegraphics[width=0.8\textwidth]{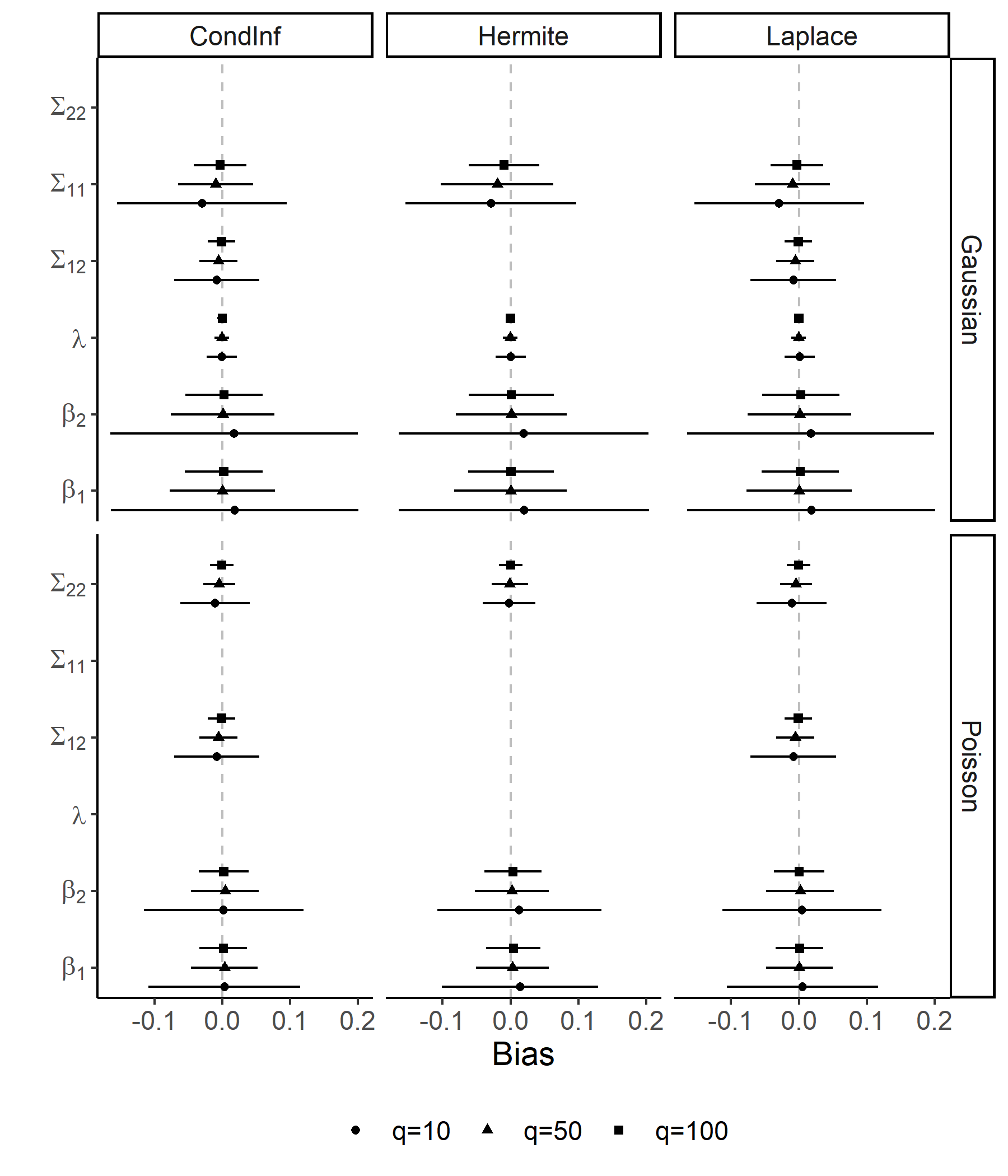}
	\caption{Estimated bias calculated from simulations of the described model for 
	three different lengths of the vectors of random components using three 
	different inference methods. The error bars show the estimated bias plus/minus 
	the estimated 
standard errors. The Hermite approximation was applied to each univariate
marginal model; therefore, there are no estimates for covariances when using this method. \label{fig:varyLevelsPlot}}
\end{figure}

\clearpage

\newpage
% ---------------------------------
% ---------------------------------
\section{Discussion \label{Sec:5}}
% ---------------------------------
% ---------------------------------

The inference method introduced in this paper extends the applicability of 
standard GLMMs in two ways: first, it allows for defining and inferring 
multivariate GLMMs, provided there exist random components representing 
clusters of observations defined in the same way in each of the marginal GLMMs; 
second, it allows to use non-Gaussian distributions for the random components. 

Remarkably, the marginal models of the defined MGLMMs can 
be of different statistical nature and at the same time represent 
complex dependence structures. Therefore, those models provide a rather flexible 
tool for 
applications. For instance, in \cite{pelck2020C} the MGLMM
contained marginal GLMMs for binomial and for Poisson distributed
responses, which appeared naturally in the process of modelling a 
system for monitoring the development of roots over time. Moreover, 
the MGLMM used in \cite{pelck2020C} could be used to detect and represent a 
first-order 
Markovian dependence induced by repeated measurements applied at 
the same experimental units over time (see also 
\cite{Shanmugam2021} for a similar application on roots 
development studies). Another example of MGLMMs including 
marginal GLMMs of different nature can be found 
in \cite{Pelck2021F}, where marginal GLMMs defined with the 
Gamma, binomial and the compound Poisson families of distributions 
were used for simultaneously modelling the development of a fungal 
infection in apples and the concentration of a series of volatile organic 
compounds, observed along time. In a third study, \cite{pelck2021D} 
used MGLMMs to simultaneously describe the students' marks obtained 
in different admission exams at the University (Gaussian distributed), 
and the performance in the course 
of geometry measured as the number of attempts required to pass the 
course (a Cox proportional model with discrete time). Those 
examples illustrate the usefulness of the MGLMMs studied in this paper. 

The inference method proposed in this paper does not involve integration of 
conditional likelihood quantities, which might be advantageous with respect to 
naive integration based methods, as illustrated in the simulation study presented in 
Section~\ref{SubSec:3.2}. The performance of the new introduced method is 
similar to the method introduced by \cite{Breslow1993}, when we assume the 
random components to be Gaussian distributed. 
Indeed, when the random components are Gaussian distributed, the inference 
functions  $\psi^*_{\bm{\beta}}$ and $\psi^*_{\bm{b}}$ 
% (given in \eqref{eq:inf1} and  \eqref{eq:inf2}) 
are similar (but not the same) to the approximate score functions  used in \cite{Breslow1993}, which are based on a Laplace approximation of the likelihood function of the GLMM $\PP$.
In this case, the inference function in \eqref{eq:inf1} is equivalent to the score 
equation of the fixed effects in \cite{Breslow1993}, whereas the inference 
function in \eqref{eq:inf2} differs from the score equation for the random effects 
by the additive term $\sigma^2 \bm{I}_{q} {\bm{b}}$, which has expectation 
zero. 
We extend the Laplace approximation method proposed by \cite{Breslow1993} 
to a multivariate context in the Appendix \ref{app:C}.

The GLMMs and MGLMMs described in this paper are constructed 
using dispersion models instead of exponential dispersion models as 
usually done in the literature of GLMMs,  see \cite{Breslow1993} and 
the literature referred there. We remark that the class of dispersion 
models defined in \cite{Jorgensen1987, Jorgensen1996} is much 
larger than the class of 
exponential dispersion models; see \cite{Cordeiro2021} and  
\cite{Labouriau2020} for a list of examples and a discussion of the 
extension of the class of dispersion models.

%\cite{Anderson2003,Breslow1993,Cordeiro2021,Jorgensen1987,
%	Jorgensen1996,Jorgensen1997,jorgensen2000,
%	Jorgensen2013,Jorgensen2012,Lee1996,Nelder1972,
%	Nelder2006,Mccullagh1989,Molas2011,Lee2001,scherner2017, 
%	Petersen2012,BarndorffNielsen2014}

% \newpage
% ---------------------------------
% ---------------------------------
\bibliography{CondInf4MGLMM3}
%\bibliography{CondInf4MGLMM}
% \bibliography{PhDbibtexLibrary}
% ---------------------------------
% ---------------------------------

\newpage
% ---------------------------------
% ---------------------------------
\appendix
\section{Appendix}
% ---------------------------------
% ---------------------------------

% ---------------------------------
% ---------------------------------
\subsection{On the identifiability of the family of conditional densities $\PP^*$\label{app:A0}}
% ---------------------------------

Here, we show that the family of conditional densities given by \eqref{eqn:2.1.3} 
is not identifiably parametrised by 
$\left ( \Bbeta,  \Bb , \lambda \right ) \in \Omega\times\R^q \times \Lambda$.
\begin{lemma}\label{LemmaA01} %%%%%%%%%%%%%%%
For any $i\in\{1, \ldots ,n\}$ and any choice of $\bm{\beta}$ ,  $\bm{b}$ and 
$\delta>0$,
there exist $\bm{\beta}_\delta \in\Omega$ such that 
$\eta_i\left ( \bm{\beta} , \bm{b} \right ) = \eta_i\left ( \bm{\beta}_\delta , \bm{b} - \delta \right )$.
\end{lemma} %%%%%%%%%%%%%%%
\begin{proof} %%%%%%%%%%%%%%%
Take arbitrary $i$, $\left ( \bm{\beta} , \bm{b} \right )$ and $\delta>0$.
Note that 
$\bm{z}_i^T \left (\bm{b} - \delta \right ) =  \bm{z}_i^T \bm{b} - \delta $
because, by construction, there is one entry of  the allocation vector $\bm{z}_i$ that is equal to one and the other entries vanish.
Assume, without loss of generality, that the first entry of the vector $\bm{x}_i$ is 
equal to $1$ (\ie, the fixed effect of the GLMM contains an intercept) so that 
$\bm{x}_i^T \bm{\beta} = \bm{\beta}_1 + \bm{\tilde {x}}_i^T 
\bm{\tilde{\beta}}$, where $\bm{\tilde {x}}_i$ and $\bm{\tilde{\beta}}$ are the 
$(k-1)$-dimensional vectors obtained by eliminating the first entry of 
$\bm{x}_i$ and $\bm{\beta}$, respectively. 
Taking 
$\bm{\beta}_\delta = \left ( \bm{\beta}_1 + \delta,  \bm{\beta}_2 , \ldots , \bm{\beta}_k \right )$ 
we have  that
\begin{align*}
\eta_i\left ( \bm{\beta} , \bm{b} \right )  
    & = \bm{x}_i^T \bm{\beta} + \bm{z}_i^T \bm{b} 
        =  \bm{\beta}_1 + \bm{\tilde {x}}_i^T \bm{\tilde{\beta}} + \bm{z}_i^T 
        \bm{b}
        =  (\bm{\beta}_1 + \delta) + \bm{\tilde {x}}_i^T   \bm{\tilde{\beta}}   + 
        \bm{z}_i^T (\bm{b} - \delta ) \\
   &  =  \bm{{x}}_i^T \bm{\beta}_\delta + \bm{z}_i^T ( \bm{b} - \delta ) 
        =  \eta_i\left ( \bm{\beta}_\delta , \bm{b} - \delta \right )
\end{align*}
The proof follows since $i$, $\left ( \bm{\beta} , \bm{b} \right )$ and $\delta$ were taken arbitrarily.
\end{proof} %%%%%%%%%%%%%%%

\newpage

% ---------------------------------
% ---------------------------------
\subsection{Technical Proofs of the Asymptotic Distribution of the Conditional Inference Based Estimates\label{app:A1}}
% ---------------------------------
In this appendix, we present a sequence of lemmas and propositions that will 
culminate with the proof of the Theorem \ref{theor:AsymptNormality}, which 
establishes consistency and joint asymptotic normality of the proposed estimator 
of $\Bbeta$ and the predictor of $\Bb$ for small values of the variance of the 
random components.

\subsubsection{Regular Inference Functions}
We recall the definition of regular inference functions used in this appendix 
for the easy of the reader
(see the details in \citeauthor{Jorgensen2012},  \citeyear{Jorgensen2012}, 
Chapter 4, from which we draw heavily).
Consider a parametric family of distributions 
${\cal P} = \{ P_\theta : \theta \in \Theta \subseteq {I\!\! R}^k \}$ 
and a $\sigma$-finite
measure $\mu$ defined on a given measurable space $({\cal X} , {\cal A )}$. For each $P_{\theta}\in \mathcal{P}$, we chose a version of the Radon-Nikodym derivative (with respect to $\mu$), denoted by
\begin{align*}
	p(\cdot;\theta)=\frac{dP_{\theta}}{d\mu}(\cdot).
\end{align*}

\begin{definition}\label{def:regular}
	A function $\Psi \, : \, \mathcal{X}\times \Theta \longrightarrow \R^k$ is said to be a regular inference function when the following conditions are satisfied for all $\theta=(\theta_1,\ldots,\theta_k)\in \Theta$ and for $i,j=1,\ldots,k$.

\begin{enumerate}[label={(\roman*)}]
	\item $\E_{\theta}[\Psi(\theta)]=0$;
	\item The partial derivative $\partial \Psi(x;\theta) / \partial \theta_i$ exists for $\mu$-almost every $x\in\mathcal{X}$;
	\item The order of integration and differentiation may be interchanged as follows:
	\begin{align*}
		\frac{\partial}{\partial \theta_i} \int_{\mathcal{X}} \Psi(x;\theta)p(x;\theta)d\mu(x) 
		= \int_{\mathcal{X}} \frac{\partial}{\partial \theta_i} [\Psi(x;\theta)p(x;\theta)]d\mu(x) \,;
	\end{align*}
\item $\E\{\psi_i(\theta)\psi_j(\theta)\}\in \R$ and the $k\times k$ matrix
\begin{align*}
	V_{\psi}(\theta)=\E\{\Psi(\theta)\Psi^T(\theta)\}
\end{align*}
is positive definite;
\item $\E\{\frac{\partial \psi_i}{\partial \theta_r}(\theta) \frac{\partial \psi_j}{\partial \theta_s}(\theta)\}\in \R $ and the $k\times k$ matrix
\begin{align*}
	S_{\psi}(\theta)=\E\{\nabla_{\theta} \Psi(\theta)\}
\end{align*}
is nonsingular.
\end{enumerate}
Here $\psi_i$ denoted the $i\tth$ component of the vector function
\begin{align*}
	\Psi(\cdot)=\left(\psi_1(\cdot), \ldots,\psi_k(\cdot)\right)^T,
\end{align*}
and $\nabla_{\theta}$ denotes the gradient operator relative to the vector $\theta$, defined by
\begin{align*}
	\nabla_{\theta}f(\theta)=\frac{\partial f}{\partial \theta^T}(\theta).
\end{align*}
\end{definition}
%%

% ---------------------------------
% ---------------------------------
\subsubsection{Some Key Lemmas}
% ---------------------------------
We denote the sequences of  roots of the inference functions  
$\psi^*_{\bm{\beta}}$ and $\psi^*_{\bm{b}}$ by  
$\{\widehat{\bm{\beta}}_n\}_{n\in \mathbb{N}}$ and  
$\{\widehat{\bm{b}}_n\}_{n\in \mathbb{N}}$  
respectively, obtained when the number of observations, $n$, increases.
Moreover, define $\Btheta = (\Bbeta  , \Bb)$ and 
$\widehat{\bm{\theta}}_n =(\widehat{\bm{\beta}}_n  ,\widehat{\bm{b}}_n )$
(for each $n\in\N$). Recall, that the inference function 
$\psi^*: \Omega\times\R^q\times\Y \rightarrow \R^{k+q}$ 
for estimating $\Btheta$ under $\PP^*$, is defined by
\begin{align*}
\psi^*(\Bbeta,\Bb) = 
\left\{\left[\psi^*_{\bm{\beta}}(\Bbeta,\Bb)\right]^T,
         \left[\psi^*_{\bm{b}}(\Bbeta,\Bb)\right]^T
         \right\}^T
        \, ,
\end{align*}
for all $\Bbeta\in\Omega$ and $\Bb\in\R^q$.

%%%%%%%%%%%%%%%%
\begin{lemma}\label{lem:psiStjerne}
Under the regularity conditions $i$-$iv$, the partial inference functions 
$\psi^*_{\bm{\beta}}$ and $\psi^*_{\Bb}$ are unbiased, that is,
\begin{align*}
		\E_{P^*_{\Bbeta,\Bb,\lambda}}
		\left [\psi^*_{\bm{\beta}}\left(\bm{\beta},\Bb;\BY \right ) \right]
		& = \bm{0}, \\
		\E_{P^*_{\Bbeta,\Bb,\lambda}}
		\left [\psi^*_{\bm{b}}\left(\bm{\beta},\Bb;\BY \right ) \right]
		& = \bm{0},
\end{align*}
for all $\bm{\beta}\in \Omega$, $\Bb \in \R^q $ and $\lambda\in \Lambda$. 
Moreover,  the partial inference functions $\psi^*_{\bm{\beta}}$ and  
$\psi^*_{\bm{b}}$, are regular.
\end{lemma}
%%%%%%%%%%%%%%%%
%%%%%%%%%%%%%%%%

%%
\begin{proof}
We show that $\psi^*_{\bm{\beta}}$ is unbiased since the unbiasedness of  
$\psi^*_{\Bb}$ follows from the same arguments. 
Take arbitrarily $\bm{\beta}\in \Omega$, $\Bb \in \R^q $ and $\lambda\in 
\Lambda$.
We aim to show that
\begin{align*}
0 = \int_{\mathcal{Y}} 
\psi^*_{\bm{\beta}}(\bm{\beta},\Bb;\bm{y})f^*(\bm{y};\bm{\beta},\Bb,\lambda)d\bm{\nu}(\bm{y}).
\end{align*} 
The regularity conditions ensure that it is allowed to interchange the order of 
differentiation and integration  in the following:
\begin{align*}
\int_{\mathcal{Y}} 
\psi^*_{\bm{\beta}}(\bm{\beta},\Bb;\bm{y})&f^*(\bm{y};\bm{\beta},\Bb,\lambda)d\bm{\nu}(\bm{y})
 \\
&=\sum_{i=1}^n  \int_{\mathcal{Y}}	 \tfrac{\partial}{\partial 
\bm{\beta}}\{d(y_i;g^{-1}(\bm{x}_i^T\bm{\beta}+\tilde{\bm{z}}_i^T\Bb))\} 
f^*(y_i;\bm{\beta},\Bb,\lambda) d\nu(y_i) \\
		&=-2\lambda\sum_{i=1}^n  \int_{\mathcal{Y}} \frac{\partial}{\partial 
		\bm{\beta}}  f^*(y_i;\bm{\beta},\Bb,\lambda) d\nu(y_i) \\
		&=-2\lambda\sum_{i=1}^n  \frac{\partial}{\partial \bm{\beta}} 
		\int_{\mathcal{Y}}  f^*(y_i;\bm{\beta},\Bb,\lambda) d\nu(y_i) 
	={0}.
\end{align*}
The proof follows since $\bm{\beta}\in \Omega$, $\Bb \in \R^q $ and $\lambda\in 
\Lambda$ are arbitrarily chosen.

The other regularity conditions for the inference functions follow straightforwardly from the assumed regularity conditions i-iv for the GLMM in play.
\end{proof}
%%%
%%%

We introduce some required notation before presenting the next lemma.
Define the sensitivity block matrices
\begin{align*}
		\begin{array}{cc}
		S_{\Bbeta\Bb} \, \, = \E[\nabla_{ \Bb} 
		\psi_{\bm{\beta}}^*(\Bbeta,\Bb;\bm{Y})], &
		S_{\Bb\Bbeta} = \E[\nabla_{ \Bbeta} 
		\psi_{\bm{b}}^*(\Bbeta,\Bb;\bm{Y})],   \\
		S_{\Bb\Bb}  = \E[\nabla_{ \Bb} 
		\psi_{\bm{b}}^*(\Bbeta,\Bb;\bm{Y})], &
		S_{\Bbeta\Bbeta} \, \,= \E[\nabla_{ \Bbeta} 
		\psi_{\bm{\beta}}^*(\Bbeta,\Bb;\bm{Y})] 	,
	\end{array}
\end{align*}
and the variability matrices
\begin{align*}
	\begin{array}{cc}
	V_{\Bbeta\Bb^*} 
\,	
=\E[\psi_{\bm{b}}^*(\Bbeta,\Bb;\bm{Y})\psi_{\bm{\beta}}^*(\Bbeta,\Bb;\bm{Y})^T],
	 &
V_{\Bb^*\Bbeta} 
=\E[\psi_{\bm{\beta}}^*(\Bbeta,\Bb;\bm{Y})\psi_{\bm{b}}^*(\Bbeta,\Bb;\bm{Y})^T],
  \\
V_{\Bb\Bb} 
=\E[\psi_{\bm{b}}^*(\Bbeta,\Bb;\bm{Y})\psi_{\bm{b}}^*(\Bbeta,\Bb;\bm{Y})^T],
&
\,\, V_{\Bbeta\Bbeta} 
\,\,=\E[\psi_{\bm{\beta}}^*(\Bbeta,\Bb;\bm{Y})\psi_{\bm{\beta}}^*(\Bbeta,\Bb;\bm{Y})^T]
 .	
	\end{array}
\end{align*}
Using these, we define
\begin{align*}
	\begin{array}{lll}
		W=D^{-1}= 	S_{\bm{b}\bm{b}}-	
		S_{\bm{\beta}\bm{b}}S_{\bm{\beta}\bm{\beta}}^{-1}S_{\bm{b}\bm{\beta}},
		 &&
		A = 
		S_{\bm{\beta}\bm{\beta}}^{-1}+S_{\bm{\beta}\bm{\beta}}^{-1}
		S_{\bm{b}\bm{\beta}}W^{-1}S_{\bm{\beta}\bm{b}}
			S_{\bm{\beta}\bm{\beta}}^{-1},\\
		E\,\, = -S_{\bm{\beta}\bm{\beta}}^{-1}S_{\bm{b}\bm{\beta}}W^{-1}, 
		&&
		C = -W^{-1}S_{\bm{\beta}\bm{b}}	S_{\bm{\beta}\bm{\beta}}^{-1}.
	\end{array}
\end{align*}
\begin{lemma}\label{appLem:Godambde}
	The inverse Godambde information for the inference function $\psi^*$ is the 
	matrix-valued function $J^{-1}_{\psi^*}: \Omega \times \R^q \rightarrow 
	\R^{(k+q)\times(k+q)}$ defined by 
	\begin{align*}
		J^{-1}_{\psi^*}=
		\begin{bmatrix}
			J_{\psi^*_{\bm{\beta}}}^{-1} & (J_{\psi^*_{\Bbeta 
			\bm{b}}}^{-1})^T \\
			 J_{\psi^*_{\Bbeta \bm{b}}}^{-1} & J_{\psi^*_{\bm{b}}}^{-1}
		\end{bmatrix},
	\end{align*}
with
	\begin{align*}
	J_{\psi^*_{\bm{\beta}}}^{-1} &= A V_{\bm{\beta}\bm{\beta}} 
	A^T+EV_{\bm{\beta}\bm{b}}A^T 
	+AV_{\bm{b}^*\bm{\beta}}E^T+EV_{\bm{b}\bm{b}}E^T\\
	J_{\psi^*_{\Bbeta \bm{b}}}^{-1}  &= C V_{\bm{\beta}\bm{\beta}} 
	A^T+DV_{\bm{\beta}\bm{b}}A^T 
	+CV_{\bm{b}\bm{\beta}}E^T+DV_{\bm{b}\bm{b}}E^T\\
	J_{\psi^*_{\bm{b}}}^{-1} &= C V_{\bm{\beta}\bm{\beta}} 
	C^T+DV_{\bm{\beta}\bm{b}^*}C^T 
	+CV_{\bm{b}\bm{\beta}}D^T+DV_{\bm{b}\bm{b}}D^T.
\end{align*}
for all $\bm{\beta}\in \Omega$ and $\Bb \in \R^q $ using the above introduced 
notation.
\end{lemma}
\begin{proof}
	The result follows from the formulas in Chapter 4 in \cite{Jorgensen2012} and 
	inversion of block matrices.
\end{proof}
\begin{lemma}\label{appLem:asympNorm}
Assume the regularity conditions $i$-$iv$.
Then, for all $\bm{\beta}\in \Omega$, $\Bb \in \R^q $ and $\lambda\in \Lambda$, 
it is true that 
\begin{align*}
	\widehat{\bm{\beta}}_n  	\xrightarrow[n\to 
	\infty]{{{P}^*_{\Bbeta,\Bb,\lambda}}} \Bbeta
	\mbox{ and }
	\widehat{\bm{b}}_n   	\xrightarrow[n\to 
	\infty]{{{P}^*_{\Bbeta,\Bb,\lambda}}}  \Bb
	\, .
\end{align*}
Moreover,
\begin{align*}
	\sqrt{n}(\widehat{\bm{\theta}}_n - \bm{\theta})\vert \bm{B}=\Bb
	\xrightarrow[n\to \infty]{\mathcal{D}} 
	N_{k+q}(\bm{0},J_{\psi^*}^{-1}(\Bbeta,\Bb)),
\end{align*}
implying that
\begin{align*}
	\sqrt{n}(\widehat{\bm{\beta}}_n - \Bbeta)\vert \bm{B}=\Bb 
	\xrightarrow[n\to 
	\infty]{\mathcal{D}} 
	N_k(\bm{0},J_{\psi^*_{\bm{\beta}}}^{-1}(\bm{\beta},\Bb))
\end{align*}
and
\begin{align*}
		\sqrt{n}(\widehat{\bm{b}}_n - \Bb)\vert \bm{B}=\Bb 
		\xrightarrow[n\to \infty]{\mathcal{D}} 
		N_q(\bm{0},J_{\psi^*_{\bm{b}}}^{-1}(\bm{\beta},\Bb))
	\, .
\end{align*}

\end{lemma}
\begin{proof}
	The proof follows from the results in Chapter 4 in 
	\cite{Jorgensen2012}, and the fact that  $\psi^*_{\bm{\beta}}$ and 
	$\psi^*_{\bm{b}}$ are regular inference functions as a consequence of 
	Lemma~\ref{lem:psiStjerne}.
\end{proof}
%%

%%%%%%%
\subsubsection{On the asymptotic variance of \texorpdfstring{$\widehat{\bm{\theta}}_n$}{$\widehat{\theta}_n$} under the family $\PP$}
%%%%%%%

%%
%%
\begin{lemma}\label{Lemma:ConsistP}
Assume the regularity conditions $i$-$iv$.
The partial solution $\{\widehat{\bm{\beta}}_n\}_{n\in \mathbb{N}}$ of 
$\psi^*_{\bm{\beta}}$ is also a solution to $\psi_{\bm{\beta}}=0$ defined in 
\eqref{eq:infFunPhi},  and the unconditionally asymptotic covariance matrices 
(for $n$ 
converging to 
infinity and $q$ fixed), denoted AV, 
of $\widehat{\bm{\beta}}_n$ and  $\widehat{\bm{b}}_n$ are given by
\begin{align}\label{eqA5-I}
	\AV (\widehat{\bm{\beta}}_n)&=	
	\E[J_{\psi^*_{\bm{\beta}}}^{-1}(\Bbeta,\bm{B})]+
	\Var [\widehat{\bm{\beta}}_n(\bm{B})] ,\\ \label{eqA5-II}
	\AV (\widehat{\bm{b}}_n)&=	
	\E[J_{\psi^*_{\bm{b}}}^{-1}(\Bbeta,\bm{B})]+
	\bm{I}_q \sigma^2  ,
\end{align}
with $\widehat{\bm{\beta}}_n(\bm{B})$ denoting the estimator of $\Bbeta$ as a 
function of $\bm{B}$ for all $n\in \mathbb{N}$, $\Bbeta\in\Omega$ and 
$\bm{B}\in \R^q$.  
Moreover, 	
\begin{align*}
	\widehat{\bm{\beta}}_n  	\xrightarrow[n\to 
	\infty]{{{P}_{\Bbeta,\lambda,\sigma^2}}} \Bbeta,
\end{align*}
for all $\bm{\beta}\in \Omega$, $\lambda\in \Lambda$ and $\sigma^2\in\R_+$.
\end{lemma}
\begin{proof}
If $\widehat{\bm{\beta}}_n$ is a solution to \eqref{eq:infFunPhi} then it is also a 
solution to \eqref{eq:inf1} when inserting  $\widehat{\bm{\beta}}_n$ and  
$\widehat{\bm{b}}_n$ for a given $n\in\mathbb{N}$.

Take  $\bm{\beta}\in \Omega$, $\lambda\in \Lambda$ and $\sigma^2\in\R_+ $ 
arbitrarily. The asymptotic covariance matrices follows from the law of total 
variance and Lemma~\ref{appLem:asympNorm}, which also implies that for all 
$\epsilon>0$
\begin{align*}
	P_{\Bbeta,\lambda,\sigma^2}(\vert \widehat{\bm{\beta}}_n - \bm{\beta}\vert > \epsilon)&= 
	\!\! \!\int_{\R^q} \!\!\!\!\! P_{\Bbeta,\Bb,\lambda}^*(\vert 
	\widehat{\bm{\beta}}_n - \bm{\beta}\vert > \epsilon \, \big\vert \, 
	\bm{B}=\bm{b}) \prod_{j=1}^q \varphi({b}_j;\sigma^2 )d\bm{b}
	\underset{n \to \infty}{\longrightarrow} 0,
\end{align*}
since $P_{\Bbeta,\Bb,\lambda}^*(\vert \widehat{\bm{\beta}}_n - \bm{\beta}\vert 
> \epsilon \, \big\vert \, \bm{B}=\bm{b})\underset{n\to 
\infty}{\longrightarrow} 0$  for all $\epsilon>0$ and $\Bb \in \R^q$. By 
the regularity 
assumptions i-iv, we can interchange the order of limit and integration. The proof 
follows since $\bm{\beta}\in \Omega$, $\lambda\in \Lambda$ and 
$\sigma^2\in\R_+$ are arbitrarily chosen.
\end{proof}

Often the distribution of the random components can be easily simulated in a computational efficient way ({\it e.g.}, when the random components are normally or t- distributed). 
In those  cases, the expectations and variances referred in \eqref{eqA5-I} and 
\eqref{eqA5-II} can be easily obtained using Monte Carlo methods (this includes 
simulations of $\bm{B}$ and calculations of estimates of $\bm{\beta}$ as a 
function of the simulated values). 

%%%%%%%%%%%%%%
\subsubsection{Proof of the Theorem \ref{theor:AsymptNormality}}\label{App:proof}
%%%%%%%%%%%%%%

The lemma below provides the calculation of the characteristic function of the 
asymptotic distribution of the sequence of estimated values of 
$\widehat{\bm{\beta}}_n$ and $\widehat{\bm{b}}_n$, which will be crucial to 
prove Theorem \ref{theor:AsymptNormality}.

%%%%%
%%%%%
\begin{lemma}\label{lemma:charFun}
Assume the regularity conditions $i$-$iv$.
There exist two random vectors $\bm{Z}_{\Bbeta}$ and $\bm{Z}_{\Bb}$ with 
characteristic functions
\begin{align*}
\E[\exp(i\bm{t_1}^T\bm{Z}_{\Bbeta})]& =
\E[	\exp(-\tfrac{1}{2}\bm{t}_1^T 
J_{\psi^*_{\bm{\beta}}}^{-1}(\Bbeta,\bm{B})\bm{t}_1 )], 
\mbox{ for all } \bm{t}_1\in \R^k ,
\\
\E[\exp(i\bm{t}_2^T\bm{Z}_{\Bb})] &=
\E[	\exp(-\tfrac{1}{2}\bm{t}_2^T 
J_{\psi^*_{\bm{b}}}^{-1}(\Bbeta,\bm{B})\bm{t}_2 )],
\mbox{ for all } \bm{t}_2\in \R^q ,
\end{align*}
respectively, such that 
\begin{align*}
     \sqrt{n} (\widehat{\bm{\beta}}_n-\Bbeta)    	\xrightarrow[n\to \infty]{\mathcal{D}}  \bm{Z}_{\Bbeta} 
      \mbox{ and }
    \sqrt{n}(\widehat{\bm{b}}_n -\Bb)    	\xrightarrow[n\to 
    \infty]{\mathcal{D}}  \bm{Z}_{\Bb}
     \, .
\end{align*}

\end{lemma}
%%%%%
%%%%%

%%%%
\begin{proof}
	By Lemma~\ref{appLem:asympNorm} we have that 
	\begin{align*}
		\sqrt{n}(\widehat{\bm{\beta}}_n -\Bbeta) \vert \bm{B}=\bm{b} 
		\overset{\mathcal{D}}{\underset{n\to \infty}{\longrightarrow}} 
		\N_k\big((\bm{0},J_{\psi^*_{\Bbeta}}^{-1}(\bm{\beta},\bm{b})\big).
	\end{align*} 
	Let $\bm{Z}_{\Bbeta}$ denote a random variable distributed according to the 
	above defined conditional asymptotically Gaussian distribution. By the 
	Portmanteau theorem the above is equivalent to
	\begin{align*}
		\E\Big[h\big(\sqrt{n}(\widehat{\bm{\beta}}_n -\Bbeta)\big)\vert 
		\bm{B}=\bm{b} \Big]  \underset{n\to\infty}{\longrightarrow} 
		\E\big[h(\bm{Z}_{\Bbeta})\vert \bm{B}=\bm{b} \big] 
	\end{align*}
	for all continuous bounded functions $h$.
	Thus, we have that
	\begin{align*}
		\E\Big[h\big(\sqrt{n}(\widehat{\bm{\beta}}_n 
		-\Bbeta)\big)\Big]&=\int_{\R^q} 
		\E\Big[h\big(\sqrt{n}(\widehat{\bm{\beta}}_n -\Bbeta)\big)\vert 
		\bm{B}=\bm{b} \Big]   
		\prod_{j=1}^q \varphi({b}_j;\sigma^2 )
		d\bm{b}
		\underset{n\to\infty}{\longrightarrow} \\
		& \int_{\R^q} \E[h(\bm{Z}_{\Bbeta})\vert \bm{B}=\bm{b} 
		]   
		\prod_{j=1}^q \varphi({b}_j;\sigma^2 )
	d\bm{b} \\
		&= \E[h(\bm{Z}_{\Bbeta})],
	\end{align*} 
	since we can interchange the order of limit and integration due to the assumed 
	regularity conditions. Therefore, we conclude that 
	\begin{align*}
		\sqrt{n}(\widehat{\bm{\beta}}_n -\Bbeta) \overset{\mathcal{D}}{\underset{n\to \infty}{\longrightarrow}} \bm{Z}_{\Bbeta}.
	\end{align*} 
	
	The characteristic  function of $\bm{Z}_{\Bbeta}$ is given by:
	\begin{align*}
		\E[\exp(i\bm{t}_1^T\bm{Z}_{\Bbeta})] &=  
		\E[\E[\exp(i\bm{t}_1^T\bm{Z}_{\Bbeta})\vert \bm{B} ]] \\
		&=\E[	\exp(-\tfrac{1}{2}\bm{t}_1^T 
		J^{-1}_{\psi^*_{\Bbeta}}(\bm{\beta},\bm{B})\bm{t}_1 )],
		\mbox{ for all } \bm{t}_1\in \R^k  .\\
	\end{align*}
The proof for $\widehat{\bm{b}}_n$ follows by similar arguments by changing 
$\widehat{\bm{\beta}}_n$ to $\widehat{\bm{b}}_n$, and $\bm{Z}_{\Bbeta}$ 
to 
$\bm{Z}_{\Bb}$ (by changing 
$J^{-1}_{\psi^*_{\bm{\beta}}}(\bm{\beta},\bm{B})$ to 
$J^{-1}_{\psi^*_{\bm{b}}}(\bm{\beta},\bm{B})$) in the above.
\end{proof}
%%%%

The theorem below corresponds to the second part of theorem \ref{theor:AsymptNormality}.
%%%
%%%
\begin{theorem} \label{prop:AsymConvApp}
Under the regularity conditions $i$-$iv$, the sequences  
$\{\hat{\bm{\beta}}_n\}_{n\in \mathbb{N}}$ and  
$\{\widehat{\bm{b}}_n\}_{n\in \mathbb{N}}$ are asymptotically Gaussian 
distributed,
when $n\to \infty$ and  $\sigma^2 \to 0+$
in the following way
%%%
\begin{align*}
		\sqrt{n}(\widehat{\bm{\beta}}_n-\Bbeta)    	\xrightarrow[\substack{n\to \infty\\ \sigma^2 \to 0+} ]{\mathcal{D}}  N_k(\bm{0},J^{-1}_{\psi^*_{\bm{\beta}}}(\bm{\beta},\bm{0})),
\end{align*}
	and
\begin{align*}
		\sqrt{n}(\widehat{\bm{b}}_n - \Bb    )	 	\xrightarrow[\substack{n\to 
		\infty\\ \sigma^2 \to 0+} ]{\mathcal{D}}  
		N_q(\bm{0},J^{-1}_{\psi^*_{\bm{b}}}(\bm{\beta},\bm{0}))
		\, .
\end{align*}
\end{theorem}
%%%%
%%%

%%
\begin{proof}
	Consider the characteristic function of $\bm{Z}_{\Bbeta}$ found in Lemma~\ref{lemma:charFun}:
\begin{align}
	\E[\exp(i\bm{t}^T\bm{Z}_{\Bbeta})] 
	&=\E[	\exp(-\tfrac{1}{2}\bm{t}^T 
	J^{-1}_{\psi^*_{\bm{\beta}}}(\bm{\beta},\bm{B})\bm{t} )],
	\mbox{ for all } \bm{t}\in \R^k 
	.\label{eq:charac}
\end{align}
Using a first order Taylor approximation, we find that
\begin{align*}
	\exp(-\tfrac{1}{2}\bm{t}^T 
	J^{-1}_{\psi^*_{\bm{\beta}}}(\bm{\beta},\bm{B})\bm{t})
	&=\exp(-\tfrac{1}{2} \sum_{i=1}^k\sum_{j=1}^k  t_it_j 
	\{J^{-1}_{\psi^*_{\bm{\beta}}}(\bm{\beta},\bm{B})\}_{ij})\\
	&=\exp(-\tfrac{1}{2} \sum_{i=1}^k\sum_{j=1}^k  t_it_j 
	\{J^{-1}_{\psi^*_{\bm{\beta}}}(\bm{\beta},\bm{0})\}_{ij})+\\
	& \quad  	\exp(-\tfrac{1}{2} \sum_{i=1}^k\sum_{j=1}^k  t_it_j \{J^{-1}_{\psi^*_{\bm{\beta}}}(\bm{\beta},\bm{0})\}_{ij}) \times \\
	& \quad -\frac{1}{2} \bm{B}^T \sum_{i=1}^k\sum_{j=1}^k t_it_j 
	\frac{\partial 
	\{J^{-1}_{\psi^*_{\bm{\beta}}}\}_{ij}}{\partial 
	\bm{b}}(\bm{\beta},\bm{0})\\
	&\quad R(\bm{B}),
	\mbox{ for all } \bm{t}\in \R^k ,
\end{align*}
where $R(\cdot)$ is the remainder term which converges to zero when 
$\bm{B}$ 
converges to zero. Thus, for $\sigma^2$ converging to 
zero, $\bm{B}$ converges to the expectation which is zero. This imply, that the 
remainder term converges to zero. Notice, that the second term has expectation 
zero since $\E[\bm{B}]=0$, so inserting the above in \eqref{eq:charac} yields
\begin{align*}
	\E_{\bm{Z}_{\Bbeta}}[\exp(i\bm{t}^T\bm{Z}_{\Bbeta})] &=  
	\exp(-\tfrac{1}{2}\bm{t}^T 
	J^{-1}_{\psi^*_{\bm{\beta}}}(\bm{\beta},\bm{0})\bm{t} 
	)+R(\bm{B})
	\underset{\sigma^2 \to 0+}{\longrightarrow}  \exp(-\tfrac{1}{2}\bm{t}^T J^{-1}_{\psi^*_{\bm{\beta}}}(\bm{\beta},\bm{0})\bm{t} ).
\end{align*}
This proves that the asymptotically distribution of  $\{\hat{\bm{\beta}}_n 
\}_{n\in \mathbb{N}}$ converges to a Gaussian distribution when $\sigma^2$ 
converges to zero. The argument for $\{\hat{\bm{b}}_n\}_{n\in 
\mathbb{N}}$ 
is equivalent and follows by changing $\hat{\bm{\beta}}_n$ to 
$\hat{\bm{b}}_n$ 
and $\bm{Z}_{\Bbeta}$ to $\bm{Z}_{\Bb}$ (changing 
$J^{-1}_{\psi^*_{\bm{\beta}}}(\bm{\beta},\bm{B})$ to 
$J^{-1}_{\psi^*_{\bm{b}}}(\bm{\beta},\bm{B})$) in the above.
\end{proof}
\newpage

% ---------------------------------
\subsection{Variance Estimation in for Models with Gaussian Random Components  \label{app:B}}
% ---------------------------------

In this section, we calculate the integral in Equation~\eqref{eq:varInt} under the 
assumption that
\begin{align*}
	g(\widehat{\bm{b}};\bm{b},\bm{\Sigma}_{\widehat{\bm{b}}})&=
	\big \vert 2\pi \bm{\Sigma}_{\widehat{\bm{b}} } \big \vert^{-\tfrac{1}{2}}
	\exp\Big(-\tfrac{1}{2} \big(\widehat{\bm{b}} 
	-\bm{b}\big)^T\bm{\Sigma}_{\widehat{\bm{b}} 
	}^{-1}\big(\widehat{\bm{b}}-\bm{b}\big)\Big) \\
	\varphi({b};\sigma^2)&= 
	(2\pi\sigma^2)^{-\tfrac{1}{2}}\exp(-\tfrac{1}{2\sigma^2}b^2).
\end{align*}
Plugging into the integral yields 
\begin{align*}
		\int_{\R^q} &
	g(\hat{\bm{b}};\bm{b},\bm{\Sigma}_{\hat{\bm{b}}})
	\varphi(\bm{b};(\bm{I}_q-\tfrac{1}{q}\bm{E}_q)\sigma^2)d\bm{b}\\
	&=\int_{\R^q} 	\big \vert 2\pi \bm{\Sigma}_{\widehat{\bm{b}} } \big 
	\vert^{-\tfrac{1}{2}}
	\exp\Big(-\tfrac{1}{2} \big(\widehat{\bm{b}} 
	-\bm{b}\big)^T\bm{\Sigma}_{\widehat{\bm{b}} 
	}^{-1}\big(\widehat{\bm{b}}-\bm{b}\big)\Big) \left \vert 
	2\pi\sigma^2\bm{I}_q\right 
	\vert^{-\tfrac{1}{2}}\exp\big(-\tfrac{1}{2}\bm{b}^T\tfrac{1}{\sigma^2}\bm{I}_q\bm{b}\big)d\bm{b}\\
	&\quad =\big \vert 2\pi \bm{\Sigma}_{\widehat{\bm{b}}_k } \big 
	\vert^{-\tfrac{1}{2}}
	(2\pi\sigma^2)^{-\tfrac{q}{2}}
	\exp\left ( -\tfrac{1}{2}\widehat{\bm{b}}^T  
	\bm{\Sigma}_{\widehat{\bm{b}}}^{-1}  \widehat{\bm{b}} \right)
	\Big \vert 2\pi 
	\big[\bm{\Sigma}_{\widehat{\bm{b}}}^{-1}+\frac{1}{\sigma^2}\bm{I}_{q}\big]^{-1}\Big
	 \vert^{\tfrac{1}{2}}\\
	& \quad \quad \exp \Big(\tfrac{1}{2}\widehat{\bm{b}}^T 
	\bm{\Sigma}_{\widehat{\bm{b}}}^{-1}    
	\big[\bm{\Sigma}_{\widehat{\bm{b}}}^{-1}+\frac{1}{\sigma^2}\bm{I}_{q}\big]^{-1}
	 \bm{\Sigma}_{\widehat{\bm{b}}}^{-1} \widehat{\bm{b}}\Big).
\end{align*}

In the case of multiple random components, we maximise the integral above for 
each random component. If the random components are nested, we only predict 
values for the random components with the highest number of clusters and then uses 
least squares to predict values for each random component, see 
Section~\ref{SubSec:2.4}. Therefore, the calculations above are changed by 
replacing $\sigma^2\bm{I}_q$ with $\sum_{j=1}^K  \sigma_{\bm{B}_j }^2 
\bm{C}_{j} \bm{C}_{j}^T$, where $\bm{B}_1,\ldots,\bm{B}_K$ denotes the 
$K\in \mathbb{N}$ nested random components, and $\bm{C}_m$ the $q\times 
q_m$ dimensional matrix specifying for each level $l$ ($l\tth$ row) which 
entry of $\bm{B}_m$ that enters the $l$th entry of $\widehat{\bm{b}}$. Here 
$q_m$ 
is the dimension of the random vector $\bm{B}_m$.

In the multivariate model described in Section~\ref{SubSec:3.1}, the above 
integral can be adapted by letting 
$\widehat{\bm{b}}^T=(\widehat{\bm{b}}_{(1)}^T,\ldots,\widehat{\bm{b}}_{(d)}^T)$
 (and thus changing the dimension of $\bm{\Sigma}_{\widehat{\bm{b}}}$) and 
replacing $\sigma^2\bm{I}_q$  with $\bm{\Sigma}\otimes\bm{I}_q$.

	%%%%%%%%%%%%%%%%%%%%%%%%%%
\subsection{Multivariate Extension of the Laplace Approximation Method
	\label{app:C}}
%%%%%%%%%%%%%%%%%%%%%%%%%%

We outline how the Laplace approximation in \cite{Breslow1993} can be 
extended to the multivariate model described in Section~\ref{SubSec:3.1}, when 
the random components follow a multivariate 
Gaussian distribution. This extension follows directly from \cite{Breslow1993} 
by redefining some matrices and vectors. We shortly describe how this was done 
in the simulation study in Section~\ref{SubSec:3.2}. The extension given below 
assumes that the marginal 
GLMMs are defined with exponential dispersion models (as in 
\cite{Breslow1993}) but this can easily be 
extended to include general dispersion models.

We assume that $\bm{B}_1,\ldots,\bm{B}_q$ are i.i.d 
according 
to a $d$-dimensional Gaussian distribution with zero mean and covariance 
matrix $\bm{\Sigma}$. Let $\bm{B}_{(j)}$ denote a vector containing all the 
$j\tth$ entries 
of $\bm{B}_1,\ldots,\bm{B}_q$ for 
$j=1,\ldots,d$.	
The above distributional assumptions implies that 
$\tilde{\bm{B}}^T=\big[(\bm{B}_{(1)})^T,\ldots,(\bm{B}_{(d)})^T\big]$ 
is Gaussian distributed 
with mean zero and covariance matrix $\bm{\Sigma} \,\otimes \, \bm{I}_q$, 
where 
$\bm{I}_q$ is the $q\times q$-dimensional identity matrix and $\otimes$ 
denotes the Kronecker product.

Recall that the $i\tth$ ($i=1,\ldots,n_j$) response in 
the $j\tth$ ($j=1,\ldots, d$) 
marginal model was denoted $y_i^{[j]}$.
Define for $j=1,\ldots,d$,  
the $n_j \times k_j$-dimensional matrix
$\bm{X}^{[j]}=[\bm{x}_{1j},\ldots,\bm{x}_{{n_j j}}]^T$, and likewise the
$n_j \times q$ matrix 
$\bm{Z}^{[j]}=[\bm{z}_{1j},\ldots,\bm{z}_{{n_j}j}]^T$. 
Based on these definitions, we define for $k=k_1+\ldots+k_d$ and 
$n=n_1+\ldots +n_d$, the
$n \times k$-dimensional matrix
$\bm{X}=\text{diag}[\bm{X}^{[1]},\ldots, \bm{X}^{[d]}]$ and the
$n \times dq$-dimensional matrix
$\bm{Z}=\text{diag}[\bm{Z}^{[1]},\ldots, \bm{Z}^{[d]}]$. Moreover, we
define for each dimension $j=1,\ldots,d$, the $n_j \times n_j$-dimensional
diagonal glm weight matrix $\bm{W}^{[j]}$ with diagonal entries
$w_{ii}^{[j]}=\tfrac{1}{2\lambda_j}\frac{2}{V_j(\mu_i^{[j]})g_j'(\mu_i^{[j]})^2}$,
and the $n\times n$ matrix 
$\bm{W}=\text{diag}[\bm{W}^{[1]},\ldots,\bm{W}^{[d]}]$.

By redefining the matrices $\bm{X}$, $\bm{Z}$, $\bm{W}$, 
$\bm{D}=\bm{\Sigma} \,\otimes \, \bm{I}_q$ and the vectors $\bm{\tilde{B}}$ 
and $\bm{y}^T=(y_1^{[1]},\ldots,y_{n_d}^{[d]})$, we can use the Laplace 
approximation in \cite{Breslow1993} to estimate the multivariate model.

\end{document}